\documentclass{IEEEoj}
\usepackage{cite}
\usepackage{graphicx, xcolor}
\usepackage{textcomp}
\def\BibTeX{{\rm B\kern-.05em{\sc i\kern-.025em b}\kern-.08em
		T\kern-.1667em\lower.7ex\hbox{E}\kern-.125emX}}
\AtBeginDocument{\definecolor{ojcolor}{cmyk}{0.93,0.59,0.15,0.02}}
\def\OJlogo{\vspace{-4pt}\hskip-4pt\includegraphics[height=18pt]{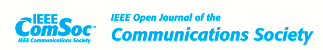}}

\usepackage{amsmath,amsfonts,amssymb}
\usepackage{mathrsfs}
\usepackage{mathtools}
\usepackage{dsfont}
\usepackage{bm}
\usepackage[long]{optidef}

\usepackage{ifthen}
\newboolean{mathfamilysansserif}
\setboolean{mathfamilysansserif}{false} %
\ifthenelse{\boolean{mathfamilysansserif}}{%
	\newcommand*{\upm}[1]{\mathsf{#1}}%
}{
	\newcommand*{\upm}[1]{\mathrm{#1}}%
}

\newcommand*{\vc}[1]{\bm{#1}}
\newcommand*{\mat}[1]{\bm{#1}}
\newcommand*{\opt}{{(\upm{opt})}}
\newcommand*{\what}[1]{\widehat{#1}}

\DeclarePairedDelimiter{\pdel}{(}{)} %
\DeclarePairedDelimiter{\bdel}{[}{]} %

\newcommand*{\idx}[1]{^{(#1)}}
\DeclarePairedDelimiter{\abs}{\lvert}{\rvert}
\DeclarePairedDelimiter{\bmat}{[}{]} %
\DeclarePairedDelimiter{\uset}{\{}{\}} %
\DeclarePairedDelimiter{\oset}{(}{)} %
\DeclareMathOperator*{\argmin}{\arg\min}
\DeclareMathOperator*{\argmax}{\arg\max}

\newcommand{\eul}{{\upm{e}}}
\newcommand{\ju}{{\upm{j}}}
\newcommand*{\onevec}{\vc{1}}
\newcommand*{\nullvec}{\vc{0}}

\newcommand*{\idmat}{\mat{I}}
\newcommand*{\nullmat}{\mat{0}}

\newcommand*\Rset{\mathbb{R}}

\newcommand*\Cset{\mathbb{C}}

\DeclarePairedDelimiter{\norm}{\lVert}{\rVert}
\newcommand*{\T}{\upm{T}}
\newcommand*{\He}{\upm{H}} %
\DeclarePairedDelimiter{\fnorm}{\lVert}{\rVert_\upm{F}}

\newcommand*{\trace}{\operatorname{Tr}\pdel}

\newcommand*{\Diag}{\operatorname{Diag}\pdel}

\newcommand*{\nullspace}{\operatorname{ker}\pdel}
\newcommand*{\rank}{\operatorname{rank}\pdel}
\newcommand*{\hmul}{\odot}

\newcommand*{\kron}{\otimes}

\newcommand*{\cnormdistr}{\mathcal{CN}\pdel}

\newcommand*{\probsym}{\mathbb{P}}

\NewDocumentCommand{\prob}{o}{
	\IfNoValueTF{#1}{\probsym\pdel*}
	{\probsym_{#1}\pdel*}
}

\NewDocumentCommand{\emprob}{o}{
	\IfNoValueTF{#1}{\what{\probsym}\pdel*}
	{\what{\probsym}_{#1}\pdel*}
}

\newcommand*{\expectsym}{\mathbb{E}}

\NewDocumentCommand{\expect}{o}{
	\IfNoValueTF{#1}{\expectsym\pdel*}
	{\expectsym_{#1}\pdel*}
}
\NewDocumentCommand{\emexpect}{o}{
	\IfNoValueTF{#1}{\what{\expectsym}\pdel*}
	{\what{\expectsym}_{#1}\pdel*}
}
\NewDocumentCommand{\var}{o}{
	\IfNoValueTF{#1}{\operatorname{Var}\pdel*}
	{\operatorname{Var}_{#1}\pdel*}
}
\NewDocumentCommand{\emvar}{o}{
	\IfNoValueTF{#1}{\what{\operatorname{Var}}\pdel*}
	{\what{\operatorname{Var}}_{#1}\pdel*}
}

\newcommand*\diff{\operatorname{d}\!}
\NewDocumentCommand{\jcby}{o}{\IfNoValueTF{#1}{\operatorname{D}\!}{\operatorname{D}_{#1}}} %

\newcommand{\cmdstkout}[1]{\ifmmode\text{\sout{\ensuremath{#1}}}\else\sout{#1}\fi}%

\usepackage{listings}
\lstdefinelanguage{NTSpython}[]{Python}{
	deletekeywords=[2]{max, range, sum, abs}, %
	morekeywords={as}
}

\usepackage{hyperref} %
\hypersetup{
	colorlinks,
	linkcolor={blue},
	citecolor={blue},
	urlcolor={blue}
}

\definecolor{thisPurple}{rgb}{0.502,0.501,0.969}
\definecolor{veryDarkGray}{rgb}{0.2,0.2,0.2}
\definecolor{darkGray}{rgb}{0.663,0.663,0.663}
\definecolor{midGray}{rgb}{0.783,0.783,0.783}
\definecolor{lightGray}{rgb}{0.902,0.902,0.902}
\definecolor{lightBlue}{rgb}{0.75,0.85,1.0}
\definecolor{darkBlue}{rgb}{0.45,0.55,1.0}
\definecolor{darkRed}{rgb}{0.961,0.426,0.411}
\definecolor{lightRed}{rgb}{0.961,0.626,0.611}
\definecolor{darkGreen}{rgb}{0.426,0.761,0.411}

\colorlet{blue1}{darkBlue!25!white}
\colorlet{blue2}{darkBlue!62!white}
\colorlet{blue3}{darkBlue}
\colorlet{blue4}{darkBlue!62!black}
\colorlet{blue5}{darkBlue!25!black}

\colorlet{red1}{darkRed!25!white}
\colorlet{red2}{darkRed!62!white}
\colorlet{red3}{darkRed}
\colorlet{red4}{darkRed!62!black}
\colorlet{red5}{darkRed!25!black}

\usepackage{xspace}
\usepackage{glossaries}
\glsdisablehyper
\usepackage{tabularx}
\usepackage{multirow}
\usepackage{booktabs}
\usepackage{algorithm}
\usepackage[noend]{algpseudocode}

\usepackage{tikz}
\usepackage{tikzscale}
\usetikzlibrary{pgfplots.groupplots} %
\usetikzlibrary{positioning,intersections,fit,shapes,arrows.meta}
\usetikzlibrary{backgrounds}
\usetikzlibrary{arrows.meta}
\usetikzlibrary{calc}
\usetikzlibrary{spy}
\usetikzlibrary{plotmarks}

\tikzset{
	>={Triangle[scale=0.7]},
	nc/.style={circle, draw=veryDarkGray, fill=lightGray, minimum width=6mm, inner sep=0.00cm, align=center, font=\footnotesize\sffamily},
	lstynw/.style={line width=0.5mm, veryDarkGray},
	lstynwhl/.style={line width=0.6mm, darkRed},
	ncm/.style={circle, draw=veryDarkGray, fill=lightGray, minimum width=0.75cm, inner sep=0.08cm, align=center, font=\footnotesize\sffamily},
	nst/.style={draw=none, align=center, font=\footnotesize},
	>={Latex[scale=0.8]},
	nr/.style={draw, rectangle, fill=lightGray, minimum height=10mm, minimum width=10mm, inner sep=1.5mm, rounded corners=0mm, align=center},
	nl/.style={draw=none, rectangle, fill=lightBlue, opacity=0.5, inner sep=1.5mm, rounded corners=2mm, align=center},
	bdot/.style={draw, circle, fill=veryDarkGray, minimum width=1mm, inner sep=0mm, align=center},
	circgray/.style={draw, circle, fill=lightGray, minimum width=8mm, inner sep=0mm, align=center},
	circred/.style={draw, circle, fill=lightRed, minimum width=8mm, inner sep=0mm, align=center, font=\footnotesize},
	txt/.style={inner sep=1mm, align=center, font=\footnotesize},
	lsty/.style={line width=0.25mm, veryDarkGray},
	boxgray/.style={nr, minimum width=8mm, minimum height=8mm},
	anlgcirc/.style={draw, circle, fill=lightGray, minimum height=4mm, inner sep=0.2mm, align=center},
}

\newcommand{\adder}[1]{
	[line cap=round, line join=round, sharp corners]
	+(0,0) circle (#1)
	+(-1*#1, 0) -- +(#1, 0)
	+(0, -1*#1) -- +(0, #1)
	+(0, 0)
}

\newcommand{\antenna}[1]{
	[line cap=round, line join=round, sharp corners]
	+(0, 0) -- +(0, 1.1*#1)
	+(0, 0.6*#1) -- +(0.5*#1, 1.1*#1)
	+(0, 0.6*#1) -- +(-0.5*#1, 1.1*#1)
	+(-0.5*#1, 1.1*#1) -- +(0.5*#1, 1.1*#1)
	+(0, 0)
}

\tikzset{deepnwnode/.style={circle, draw=none, fill=lightRed, minimum height=0, minimum width=0}}

\makeatletter
\DeclareRobustCommand{\tikzvdots}{%
	\vbox{
		\baselineskip4\p@\lineskiplimit\z@
		\kern1\p@
		\hbox{.}\hbox{.}\hbox{.}
}}
\makeatother

\usetikzlibrary{external}
\tikzset{external/system call={pdflatex \tikzexternalcheckshellescape -halt-on-error -interaction=batchmode --extra-mem-bot=2000000000 --extra-mem-top=2000000000 -jobname "\image" "\texsource"}}

\usepackage{pgfplots}
\pgfplotsset{compat=1.18}

\pgfplotscreateplotcyclelist{markcycle}{mark = o, mark = triangle, mark = diamond, mark = square, mark = pentagon, mark = star, mark = oplus, mark = otimes}

\pgfplotsset{
	minor grid style={dotted},
	major grid style={solid},
	grid = major, 
	minor tick num=4,
	xlabel near ticks,
	ylabel near ticks,
	every axis post/.style={
		xlabel style={align=center}, 
		ylabel style={align=center},
	},
	every axis plot/.append style={
		thick, 
		mark options=solid, 
	},
	legend style={at={(1,1.02)}, draw=none, fill=none, font=\small, anchor=south east, /tikz/every even column/.append style={column sep=0.2cm}},
}

\usepackage{amsthm}
\newtheorem{theorem}{Theorem}[section]
\newtheorem{lemma}[theorem]{Lemma}

\newtheorem{assumption}{Assumption}[section]
\newtheorem*{remark}{Remark}

\usepackage{siunitx}
\DeclareSIUnit{\dB}{dB}
\DeclareSIUnit{\kkilo}{\kilo\relax}

\makeatletter
\DeclareRobustCommand\onedot{\futurelet\@let@token\@onedot}
\def\@onedot{\ifx\@let@token.\else.\null\fi\xspace}

\def\eg{e.g\onedot}

\def\ie{i.e\onedot}

\def\cf{c.f\onedot}

\def\wrt{w.r.t\onedot}

\def\wrt{w.r.t\onedot}

\makeatother

\newcommand*{\matel}[2]{[#1]_{#2}}

\newcommand*{\vect}{\operatorname{vec}\pdel}
\newcommand*{\vectd}{\operatorname{vec}_{\upm{d}}\pdel}
\newcommand*{\vectr}{\operatorname{vec}_{\upm{re}}\pdel}
\newcommand*{\vecti}{\operatorname{vec}_{\upm{im}}\pdel}

\newcommand*{\Blkdiag}{\operatorname*{Blkdiag}\pdel}

\newcommand*{\vxi}{\vc{\xi}}

\newcommand*{\va}{\vc{a}}
\newcommand*{\vb}{\vc{b}}

\newcommand*{\vf}{\vc{f}}
\newcommand*{\vg}{\vc{g}}
\newcommand*{\vh}{\vc{h}}
\newcommand*{\vn}{\vc{n}}
\newcommand*{\vp}{\vc{p}}
\newcommand*{\vq}{\vc{q}}
\newcommand*{\bs}{\vc{s}}

\newcommand*{\vw}{\vc{w}}
\newcommand*{\vx}{\vc{x}}
\newcommand*{\vy}{\vc{y}}

\newcommand*{\vtheta}{\vc{\theta}}

\newcommand*{\vlambda}{\vc{\lambda}}
\newcommand*{\vzeta}{\vc{\zeta}}
\newcommand*{\vpsi}{\vc{\psi}}

\newcommand*{\bA}{\mat{A}}
\newcommand*{\bB}{\mat{B}}
\newcommand*{\bC}{\mat{C}}

\newcommand*{\bE}{\mat{E}}

\newcommand*{\bR}{\mat{R}}
\newcommand*{\bS}{\mat{S}}

\newcommand*{\bX}{\mat{X}}

\newcommand*{\bZ}{\mat{Z}}

\newcommand*{\hbC}{\widehat{\bC}}
\newcommand*{\hbR}{\widehat{\bR}}

\newcommand*{\hvh}{\widehat{\vh}}
\newcommand*{\hvpsi}{\widehat{\vpsi}}

\newcommand*{\tvh}{\widetilde{\vh}}
\newcommand*{\tbR}{\widetilde{\bR}}

\newcommand*{\bva}{\breve{\va}}
\newcommand*{\brs}{\breve{s}}
\newcommand*{\bvb}{\breve{\vb}}

\newcommand*{\bvzeta}{\breve{\vzeta}}

\newcommand*{\uvb}{\underline{\vb}}
\newcommand*{\uvf}{\underline{\vf}}
\newcommand*{\ubvb}{\underline{\bvb}}
\newcommand*{\uvg}{\underline{\vg}}
\newcommand*{\uvzeta}{\underline{\vzeta}}
\newcommand*{\uvpsi}{\underline{\vpsi}}

\newcommand*{\ubLambda}{\underline{\bLambda}}

\newcommand*{\ubvzeta}{\underline{\breve{\vzeta}}}
\newcommand*{\uhvpsi}{\underline{\what{\vpsi}}}

\newcommand*{\uhvpsid}{\uhvpsi_{\upm{d}}}
\newcommand*{\uhvpsir}{\uhvpsi_{\upm{re}}}
\newcommand*{\uhvpsii}{\uhvpsi_{\upm{im}}}

\newcommand*{\uhbPsi}{\underline{\what{\bPsi}}}

\newcommand*{\bGamd}{\bGamma_{\upm{d}}}
\newcommand*{\bGamr}{\bGamma_{\upm{re}}}
\newcommand*{\bGami}{\bGamma_{\upm{im}}}
\newcommand*{\bGamein}{\bGamma_{\alpha}}

\newcommand*{\bGamma}{\mat{\Gamma}}
\newcommand*{\bTheta}{\mat{\Theta}}
\newcommand*{\bPhi}{\mat{\Phi}}
\newcommand*{\bPsi}{\mat{\Psi}}
\newcommand*{\bLambda}{\mat{\Lambda}}
\newcommand*{\hbPsi}{\widehat{\bPsi}}

\newcommand*{\cA}{\mathcal{A}}
\newcommand*{\cB}{\mathcal{B}}
\newcommand*{\cF}{\mathcal{F}}

\newcommand*{\cQ}{\mathcal{Q}}
\newcommand*{\cS}{\mathcal{S}}

\newcommand*{\cbook}{\cA_\upm{rf}}

\newcommand*{\numcb}{A}
\newcommand*{\numfbv}{N_\upm{fb}} %
\newcommand*{\numrf}{M_\upm{rf}}
\newcommand*{\numaly}{L_\upm{rf}}
\newcommand*{\numgnnly}{L_\upm{gcn}}
\newcommand*{\numgnnfilt}{F_\upm{gcn}}

\newcommand*{\weightvec}{\vw}

\newcommand*{\prdf}{\mathcal{P}\pdel}
\newcommand*{\supp}{\operatorname{supp}}
\newcommand*{\relu}{\operatorname{ReLU}}
\newcommand*{\cM}{\mathcal{M}}

\newcommand*{\ssample}{\mathcal{S}}
\newcommand*{\dataset}{\mathcal{D}}
\newcommand*{\trainset}{\dataset_\upm{train}}
\newcommand*{\valset}{\dataset_\upm{val}}

\newcommand*{\estqtl}{\what{\cQ}}

\newcommand*{\stepfun}{u}
\newcommand*{\apxstepfun}{\widetilde{u}}
\newcommand*{\sinrdl}{\operatorname{SINR}^\upm{Dl}}
\newcommand*{\sinrul}{\operatorname{SINR}^\upm{Ul}}
\newcommand*{\sinrdlest}{\what{\operatorname{SINR}}^\upm{Dl}}
\newcommand*{\sinrulest}{\what{\operatorname{SINR}}^\upm{Ul}}

\newcommand*{\pout}{P_{\upm{out}}}
\newcommand*{\estpout}{\what{P}_{\upm{out}}}
\newcommand*{\pmax}{P_{\upm{max}}}
\newcommand*{\ppilot}{P_{\upm{pil},i}}

\newcommand*{\gcnn}{\bPhi_{\upm{gcn}}}

\newcommand*{\gcnnly}{\breve{\bPhi}_{\upm{gcn}}}

\newcommand*{\ocplx}{\mathcal{O}\pdel}
\newcommand*{\digbfmap}{\mathcal{F}_\upm{Dl}}
\newcommand*{\digbfmapul}{\mathcal{F}_\upm{Ul}}
\newcommand*{\realulfun}{\uvf_\upm{Ul}}
\renewcommand*{\opt}{\star}

\makeatletter
\DeclareRobustCommand{\cvdots}{%
	\vbox{
		\baselineskip4\p@\lineskiplimit\z@
		\kern-\p@
		\hbox{.}\hbox{.}\hbox{.}
}}
\makeatother

\newacronym{ad}{AD}{anomaly detection}
\newacronym{roc}{ROC}{receiver operating characteristic}
\newacronym{auc}{AUC}{area under the curve}

\newacronym{psd}{PSD}{positive semidefinite}
\newacronym{svd}{SVD}{singular value decomposition}
\newacronym{pca}{PCA}{principal component analysis}
\newacronym{rpca}{RPCA}{robust principal component analysis}

\newacronym{snle}{SNLE}{system of nonlinear equations}
\newacronym{bcd}{BCD}{block coordinate descent}
\newacronym{bsca}{BSCA}{block successive convex approximation}
\newacronym{mip}{MIP}{Mixed-Integer Program}
\newacronym{kkt}{KKT}{Karush-Kuhn-Tucker}
\newacronym{pgd}{PGD}{projected gradient descent}

\newacronym{dn}{DN}{deep network}
\newacronym{dl}{DL}{deep learning}
\newacronym{nn}{NN}{neural network}
\newacronym{gnn}{GNN}{graph neural network}
\newacronym{gcnn}{GCN}{graph convolutional neural network}
\newacronym{lista}{LISTA}{robust principal component analysis}
\newacronym{mlp}{MLP}{multilayer perceptron}
\newacronym{sgd}{SGD}{stochastic gradient descent}

\newacronym{fdd}{FDD}{frequency division duplex}
\newacronym{mu}{MU}{multi-user}
\newacronym{miso}{MISO}{multiple-input single-output}
\newacronym{mimo}{MIMO}{multiple-input multiple-output}
\newacronym{mumiso}{MU MISO}{multi-user multiple-input single-output}
\newacronym{csi}{CSI}{channel state information}
\newacronym{snr}{SNR}{signal-to-noise ratio}
\newacronym{sinr}{SINR}{signal-to-interference-plus-noise ratio}
\newacronym{cb}{CB}{codebook}
\newacronym{noma}{NOMA}{non-orthogonal multiple access}
\newacronym{qos}{QoS}{quality of service}
\newacronym{bs}{BS}{base station}
\newacronym{ue}{UE}{user equipment}
\newacronym{rf}{RF}{radio frequency}
\newacronym{dft}{DFT}{discrete Fourier transform}
\newacronym{bf}{BF}{beamforming}

\begin{document}
	
\receiveddate{7 January 2026}
\reviseddate{9 February 2026}
\accepteddate{24 February, 2026}
\publisheddate{5 March, 2026}
\currentdate{18 March, 2026}
\doiinfo{10.1109/OJCOMS.2026.3671023}
\title{Hybrid Downlink Beamforming with Outage Constraints under Imperfect CSI using Model-Driven Deep Learning}
	
\author{~Lukas~Schynol\IEEEauthorrefmark{1}\IEEEmembership{(Graduate Student Member, IEEE)} and Marius~Pesavento\IEEEauthorrefmark{1} \IEEEmembership{(Senior Member, IEEE)}}
\affil{Technical University of Darmstadt, Darmstadt, Germany}
\corresp{CORRESPONDING AUTHOR: Lukas~Schynol (e-mail: lschynol@nt.tu-darmstadt.de).}
\authornote{This work was financially supported by the Federal Ministry of Research, Technology and Space of Germany in the project "Open6GHub+" (grant no. 16KIS2407).\\[-32pt]}
\markboth{Submission to OJCOMS}{Schynol and Pesavento}

\begin{abstract}
	We consider energy-efficient multi-user hybrid downlink \gls{bf} and power allocation under imperfect \gls{csi} and probabilistic outage constraints.
	In this domain, classical optimization methods resort to computationally costly conic optimization problems.
	Meanwhile, generic \gls{dn} architectures lack interpretability and require large training data sets to generalize well.
	In this paper, we therefore propose a lightweight model-aided deep learning architecture based on a greedy selection algorithm for analog beam codewords.
	The architecture relies on an instance-adaptive augmentation of the signal model to estimate the impact of the \gls{csi} error. 
	To learn the \gls{dn} parameters, we derive a novel and efficient implicit representation of the nested constrained \gls{bf} problem and prove sufficient conditions for the existence of the corresponding gradient.
	In the loss function, we utilize an annealing-based approximation of the outage compared to conventional quantile-based loss terms.
	This approximation adaptively anneals towards the exact probabilistic constraint depending on the current level of \gls{qos} violation.
	Simulations validate that the proposed \gls{dn} can achieve the nominal outage level under \gls{csi} error due to channel estimation and channel compression, while allocating less power than benchmarks.
	Thereby, a single trained model generalizes to different numbers of users, \gls{qos} requirements and levels of \gls{csi} quality.
	We further show that the adaptive annealing-based loss function can accelerate the training and yield a better power-outage trade-off.
\end{abstract}
\glsresetall
\begin{IEEEkeywords}
	Deep unrolling with constraints, hybrid beamforming, multi-user MISO, probabilistic outage constraint, mixed-integer, implicit function differentiation
\end{IEEEkeywords}
	
	\maketitle

\section{Introduction}\label{sec:introduction}
Downlink \gls{bf} serves as a fundamental technology in modern wireless networks, including 3GPP 5G \cite{heathjrFoundationsMIMOCommunication2018} and the upcoming 6G standard \cite{you6GWirelessCommunication2021}.
Nonetheless, large-scale multi-antenna systems, network densification and surging traffic demands led to an increased energy usage and cost \cite{iEnergyefficient5GGreener2020, cheng5GNetworkDeployment2022}.
As such, energy efficiency is of marked interest in future wireless systems \cite{jiangRoad6GComprehensive2021}.
\par
Considering these challenges, hybrid \gls{rf} front-ends that integrate low-dimensional digital processing with high-dimensional analog processing offer a compelling advantage.
Hybrid \gls{bf} architectures improve energy consumption and cost while achieving a spatial separation and throughput comparable to fully digital architectures \cite{molischHybridBeamformingMassive2017, sohrabiHybridDigitalAnalog2016}.
However, restrictions on the configuration of analog components make it challenging to optimize beam patterns \wrt network utilities such as sum-rate or allocated power under \gls{qos} requirements in \gls{mu} systems \cite{shlezingerArtificialIntelligenceEmpoweredHybrid2024}.
The problem is intensified in practical systems with a large number of antennas, \eg, systems operating millimeter wave or the upcoming 6G FR3 bands, where \gls{csi} degradation arises from factors such as estimation errors, quantization of feedback, time-varying fading or space and frequency subsampling of the channel \cite{shlezingerArtificialIntelligenceEmpoweredHybrid2024, joudehRobustTransmissionDownlink2016}.
Classical approaches commonly tackle the resulting robust downlink \gls{bf} problem by regularization and diagonal loading \cite{bjornsonOptimalMultiuserTransmit2014, nguyenMMSEPrecodingMultiuser2014, bjornsonMassiveMIMOImperfect2016}.
However, these approaches are not easily applicable in multi-user downlink scenarios with \gls{qos} constraints.
Alternatively, robust methods based on conic optimization, which require knowledge of the error distribution and are often computationally costly, have been explored, \eg, \cite{joudehRobustTransmissionDownlink2016, wangOutageConstrainedRobust2014, botrosshenoudaProbabilisticallyConstrainedApproachesDesign2008, wajidRobustDownlinkBeamforming2013, xuRobustPowerControl2015, vucicRobustQoSConstrainedOptimization2009}.
\par
In recent years, \gls{dl} has thus been considered to address \gls{bf} and power allocation in a computationally efficient manner \cite{zhangDeepLearningMobile2019}.
To mitigate concerns regarding the black-box nature of generic \gls{dn} architectures and their generalization capability \cite{shafinArtificialIntelligenceEnabledCellular2020}, recent work turns to the concept of model-aided \gls{dl} and deep unrolling \cite{gregorLearningFastApproximations2010, shlezingerModelBasedDeepLearning2022a}.
By integrating existing domain knowledge and modifying classical algorithms with learnable components, model-aided \gls{dn} architectures offer advantages in training data efficiency, generalization capability and explainability \cite{schynolCoordinatedSumRateMaximization2023}.
This motivates us to tackle energy-efficient and robust hybrid \gls{bf} from the perspective of model-aided \gls{dl}.

\subsection{Related Work}\label{sec:relwork}
Hybrid \gls{bf} is extensively discussed in the literature \cite{molischHybridBeamformingMassive2017}.
The majority of existing non-\gls{dl} approaches decouple \gls{rf} and baseband beamformers and either aim to maximize the sum-rate under power constraints, \eg, \cite{sohrabiHybridDigitalAnalog2016, wangHybridPrecoderCombiner2018}, minimize the MMSE or reconstruct a desired beam pattern, \eg, \cite{ayachSpatiallySparsePrecoding2014, alkhateebChannelEstimationHybrid2014}, see also \cite{elbirTwentyFiveYearsAdvances2023}.
\par
Energy-efficient hybrid \gls{bf} realized by power minimization under \gls{qos} constraints is investigated in \cite{mallaTransmissionStrategiesMultiuser2017, tajallifarQoSAwareHybridBeamforming2021, chenHybridBeamformingData2022, kimHybridBeamformingBased2024, vazquezMultiuserDownlinkHybrid2017, zangPartiallyConnectedHybridBeamforming2020, wenQoSGuaranteedHybridBeamforming2023}.
In \cite{mallaTransmissionStrategiesMultiuser2017, tajallifarQoSAwareHybridBeamforming2021, chenHybridBeamformingData2022} and \cite[Alg.~2]{kimHybridBeamformingBased2024}, the \gls{mu} \gls{miso} and \gls{mu} \gls{mimo} problem is tackled by zero-forcing precoding and block diagonalization.
Under certain conditions, block diagonalization and zero-forcing fully eliminate interference between user signals with digital precoding, thereby simplifying the problem substantially. 
However, when \gls{csi} uncertainty is present, full cancellation cannot be guaranteed, resulting in a degraded \gls{qos} level.
Instead of block diagonalization, the algorithms in \cite{vazquezMultiuserDownlinkHybrid2017, zangPartiallyConnectedHybridBeamforming2020,kimHybridBeamformingBased2024} and \cite[Alg.~1]{wenQoSGuaranteedHybridBeamforming2023} address \gls{qos}-constrained power minimization with a block-iterative approach, where the nested subproblems rely on conic and semidefinite programming. 
In \cite{yePowerMinimizationHybrid2016, hegdeHybridBeamformingLargescale2017}, the analog beams are restricted to a discrete set of codewords, enabling a particularly cost-efficient implementation.
Algorithm~2 of \cite{hegdeHybridBeamformingLargescale2017} employs a greedy correction strategy for codeword selection, whereas \cite{yePowerMinimizationHybrid2016} relaxes the codeword selection problem into a cone program with sparse regularization.
Generally speaking, the corresponding optimization problems of all aforementioned methods assume perfect \gls{csi}.

\par
To the best of our knowledge, the only non-\gls{dl}-based work addressing robust \textit{and} energy-efficient hybrid \gls{bf} is \cite{tajallifarRobustFeasibleQoSAware2024}, which investigates \gls{mu} \gls{mimo} \gls{bf} using a worst-case framework in two variants.
The first applies a cutting-set method that iteratively solves conic optimization problems to construct a set of worst-case channels.
The second reduces the complexity by leveraging block diagonalization, leading to a more conservative allocation of resources.
\par
Considering these trade-offs, model-aided \gls{dl} methods present an opportunity for robust \gls{bf}.
Most literature on fully digital or hybrid \gls{bf} using \gls{dl} focuses on sum-rate maximization under power constraints, since power constraints are straightforward to implement by a projection.
In the aforementioned, robustness is introduced by training with imperfect \gls{csi} data, \eg, \cite{huIterativeAlgorithmInduced2021, jinModelDrivenDeepLearning2023}, by sample averaging \cite{jinLowComplexityJointBeamforming2024}, or by addressing uncertainty via the usage of statistical \gls{csi} \cite{liuRobustDownlinkPrecoding2022, shiDeepLearningBasedRobust2021}.
A projection approach, however, is not suitable for most \gls{qos} constraints due to their nonconvexity or probabilistic formulation.
In \cite{xiaDeepLearningFramework2020}, a model-aided \gls{dl} framework leveraging uplink-downlink duality for \gls{qos}-constrained fully-digital \gls{bf} is proposed.
The authors design a \gls{dn} that predicts the power allocation, which enables the recovery of precoders by the analytic solution structure, but only perfect \gls{csi} is considered.
In \cite{laviLearnRapidlyRobustly2023, wangRobustHybridBeamforming2025}, sum-rate maximization in robust hybrid \gls{bf} without \gls{qos} constraints is tackled by a worst-case max-min approach using unrolled \gls{pgd}.
The authors of \cite{liuBiLevelDeepUnfolding2024} extend the work in \cite{laviLearnRapidlyRobustly2023, wangRobustHybridBeamforming2025} for fully digital \gls{bf} in an integrated sensing and communication system with \gls{qos} constraints.
Thereby, a probabilistic outage constraint is converted into a worst-case error constraint and integrated into the unrolled \gls{dn} architecture via the Lagrangian formulation. %
However, the accuracy of constraint satisfaction during testing remains uncertain.
In \cite{shresthaOptimalSolutionsJoint2023}, worst-case robust energy-efficient \gls{bf} and antenna selection with a branch-and-bound algorithm and semidefinite programming is accelerated by branch selection using a \gls{gnn}. %
\par
\Gls{dl} for \gls{bf} that incorporate probabilistic outage constraints into the training objective are considered in \cite{psomasDesignAnalysisSWIPT2022, yingDeepLearningBasedJoint2024, jangDeepLearningApproach2022, youDataAugmentationBased2021}. 
In \cite{psomasDesignAnalysisSWIPT2022, yingDeepLearningBasedJoint2024, jangDeepLearningApproach2022}, generic \gls{dn} architectures learn robust \gls{bf} policies by embedding the outage constraint into the training objective as a penalty.
Specifically, \cite{psomasDesignAnalysisSWIPT2022} and \cite{yingDeepLearningBasedJoint2024} employ a quantile-based formulation of the constraints for fully digital \gls{mu}-\gls{miso} downlink \gls{bf}, while \cite{jangDeepLearningApproach2022} applies a fixed approximation of the outage probability for power allocation in multiple access.
In contrast, \cite{youDataAugmentationBased2021} combines the quantile-based approach with primal-dual stochastic gradient descent/ascent for robust fully-digital \gls{bf}, thereby avoiding the need to predefine penalty weights.
In \cite{liangDataModelDrivenDeep2024}, a model-aided \gls{dn} is trained to maximize a nominal quantile of the sum-rate, after which the maximum transmit power is bisected until the target \gls{qos} is achieved.
The method does not support instance-dependent \gls{qos} constraints and the inclusion of differing \gls{qos} levels is not straightforward.
We further observe that the approaches in \cite{psomasDesignAnalysisSWIPT2022, yingDeepLearningBasedJoint2024, jangDeepLearningApproach2022, youDataAugmentationBased2021, liangDataModelDrivenDeep2024} are only verified on systems with small arrays or low \gls{sinr} requirements ($\leq \qty{0}{\decibel}$).

\subsection{Contributions}
Motivated by the drawbacks of \cite{psomasDesignAnalysisSWIPT2022, yingDeepLearningBasedJoint2024, jangDeepLearningApproach2022, youDataAugmentationBased2021, liangDataModelDrivenDeep2024}, we study energy-efficient \gls{mu} hybrid downlink \gls{bf} under probabilistic \gls{qos} constraints with model-aided \gls{dl}, which, to our knowledge, has not been previously investigated.
As in \cite{yePowerMinimizationHybrid2016, hegdeHybridBeamformingLargescale2017, heCodebookBasedHybridPrecoding2017}, but in contrast to the classical method for hybrid \gls{bf} with \gls{qos} constraints under imperfect \gls{csi} in \cite{tajallifarRobustFeasibleQoSAware2024}, the columns of the analog beamformer are restricted to a discrete \gls{cb}. %
Our contributions can be summarized as follows:
\begin{itemize}
	\item Leveraging model-aided \gls{dl}, we propose a novel \gls{dn} architecture based on uplink-downlink duality for power minimization and precoding in hybrid \gls{bf} under outage probability constraints.
	In contrast to the framework in \cite{xiaDeepLearningFramework2020} and similar succeeding works, we do not predict the power allocation and instead rely on an adaptive lightweight mapping of \gls{csi} using \glspl{gnn}, thereby avoiding the second-order cone programming or semidefinite programming of classical approaches.
	Compared to previous works on \gls{qos}-constrained \gls{bf} under imperfect \gls{csi} \cite{psomasDesignAnalysisSWIPT2022, yingDeepLearningBasedJoint2024, jangDeepLearningApproach2022, youDataAugmentationBased2021, liangDataModelDrivenDeep2024, liuBiLevelDeepUnfolding2024, shresthaOptimalSolutionsJoint2023}, the resultant models can adapt to various channels conditions with changing \gls{csi} quality, required \gls{sinr} levels, the number of users, and supports instantaneous and statistical \gls{csi}.
	Crucially, our method does not rely on knowledge of the analytical \gls{csi} error distribution as in \cite{liuBiLevelDeepUnfolding2024, shresthaOptimalSolutionsJoint2023}.
	\item In the process, we derive a novel and efficient implicit representation for the underlying uplink \gls{bf} and power minimization problem under \gls{qos} constraints and the corresponding gradient \wrt the \gls{csi}.
	We further prove sufficient conditions for the existence of the gradient.
	\item We propose an annealing-based approach that, unlike the fixed approximation used in \cite{jangDeepLearningApproach2022}, employs an adaptive annealing coefficient to integrate probabilistic constraints into \gls{dn} training.
	The effectiveness of this method is demonstrated in comparison to conventional quantile-based implementations as in \cite{psomasDesignAnalysisSWIPT2022, yingDeepLearningBasedJoint2024, youDataAugmentationBased2021}.
	\item We provide an extensive empirical comparison of the power-outage trade-off of our proposed method against benchmark algorithms.
	We further verify the generalization and adaptation capabilities of our proposed method under varying \gls{csi} quality, \gls{qos} requirements and number of users in the network.
\end{itemize}

\par
Compared to our preliminary work in \cite{schynolCodebookBasedDownlinkBeamforming2024} that examines purely codebook-based \gls{bf}, this work considers hybrid \gls{bf} instead, thereby providing a novel representation of the underlying uplink \gls{bf} problem to enable an efficient gradient computation during \gls{dn} learning.
In addition, we extensively validate our method with 5-fold cross validation, thereby proposing a convergence metric for the constrained training problem to ensure fairness and consistency between runs.
We further provide empirical comparisons between the annealing-based and quantile-based probabilistic constraint implementation. %

\subsection{Paper Structure}
The remainder of the paper is structured as follows.
In Sec.~\ref{sec:model_and_task}, we specify the system model and define the investigated optimization problem.
Building on the greedy hybrid \gls{bf} approach \cite{hegdeHybridBeamformingLargescale2017} reviewed in Sec.~\ref{sec:greedybf_perfcsi}, we present the proposed architecture and its gradient based on a novel implicit representation in Sec.~\ref{sec:prop_robust_greedybf}.
In particular, we discuss training with the proposed annealing-based loss function and provide a technical comparison to the prevailing approach for chance-constrained \gls{qos} constraints.
Sec.~\ref{sec:results} discusses empirical results and Sec.~\ref{sec:conclusion} concludes this paper.

\subsection{Notation}
Scalars, vectors and matrices are denoted as $x$, $\vx$ and $\bX$.
$\matel{\vx}{i}$ and $\matel{\vx}{i,j}$ represent the $i$th and $(i,j)$th element, respectively.
The notation of functions follows the elements of their co-domain.
$\matel{\bX}{i,:}$ and $\matel{\bX}{:, j}$ are the vectors resulting from slicing the $i$th row or $j$th column, while $\matel{\vx}{\neg i}$, $\matel{\bX}{\neg i, :}$ or $\matel{\bX}{:, \neg j}$ are the vector or matrix with the $i$th element, row or column excluded.
Hadamard and Kronecker products are denoted by $\hmul$ and $\kron$.
We define $\prob[]{\cdot}$, $\emprob[]{\cdot}$, $\expect[]{\cdot}$ and $\emexpect[]{\cdot}$ as probability, empirical probability, expected value and sample mean, respectively.
The Jacobi matrix of a vector-valued function $\vf(\vx)$ is denoted as $\jcby[\vx] \vf (\vx)$.
Finally, an underline $\underline{\vx}$ indicates the real-valued version of a complex variable $\vx$.

\section{System Model and Problem Formulation}\label{sec:model_and_task}
\subsection{System Model}\label{sec:system_model}
\begin{figure}
	\newcommand*{\xrf}{1.5}
\newcommand*{\xanlg}{\xrf+1.2}
\newcommand*{\xadd}{\xanlg+0.75}
\newcommand*{\xant}{\xadd+0.5}
\newcommand*{\xuant}{\xant+1.75}
\newcommand*{\xuadd}{\xuant+0.6}
\newcommand*{\xoutgplt}{\xuant+0.75}
\newcommand*{\yrf}{1.2}
\newcommand*{\yue}{1.4}
\newcommand*{\ydesc}{-1.95}

\tikzsetnextfilename{system_model}
\begin{tikzpicture}
	\node[boxgray, minimum height=3.6cm] (bb) at (0.0, 0.0) {$\bB$};
	\draw[lsty, <-] (bb.west) -- ++(-0.5, 0) node[txt, above] {$\oset{\breve{s}_i}_i$}; %
	
	\node[boxgray, minimum height=1.2cm] (rf1) at (\xrf, \yrf) {RF\\Chain};
	\node[boxgray, minimum height=1.2cm] (rfN) at (\xrf, -\yrf) {RF\\Chain};
	\node[txt] (rfdot) at (\xrf, 0) {\tikzvdots};
	
	\draw[lsty] ([yshift=1.2cm]bb.east) -- ([yshift=0mm]rf1.west);
	\draw[lsty] ([yshift=-1.2cm]bb.east) -- ([yshift=0mm]rfN.west);

	\node[anlgcirc, minimum width=4mm] (a11) at (\xanlg, \yrf+0.4) {$\times$};
	\node[anlgcirc, minimum width=4mm] (aM1) at (\xanlg, \yrf-0.4) {$\times$};
	\node[txt] (rfdot) at (\xanlg, 1.2) {\scalebox{0.6}{\textbf{\tikzvdots}}};
	
	\node[anlgcirc, minimum width=4mm] (a1M) at (\xanlg, -\yrf+0.4) {$\times$};
	\node[anlgcirc, minimum width=4mm] (aMM) at (\xanlg, -\yrf-0.4) {$\times$};
	\node[txt] (rfdot) at (\xanlg, -1.2) {\scalebox{0.6}{\textbf{\tikzvdots}}};
	
	\begin{scope}[on background layer]
		\node[nl, fill=midGray, fit=(a11)(aM1), label={[label distance=-0.6mm]above:$\va_1$}] (a1) {};
		\node[nl, fill=midGray, fit=(a1M)(aMM), label={[label distance=-0.6mm]above:$\va_{\numrf}$}] {};
	\end{scope}
	
	\draw[lsty] ([yshift=0.4cm]rf1.east) -- ([yshift=0mm]a11.west);
	\draw[lsty] ([yshift=-0.4cm]rf1.east) -- ([yshift=0mm]aM1.west);
	\draw[lsty] ([yshift=0.4cm]rfN.east) -- ([yshift=0mm]a1M.west);
	\draw[lsty] ([yshift=-0.4cm]rfN.east) -- ([yshift=0mm]aMM.west);

	\draw[lsty] (a11) -- (\xadd+0.1, \yrf+0.4) \adder{0.15} -- (\xant, \yrf+0.4) \antenna{0.5};
	\draw[lsty] (a1M) -| (\xadd+0.1, \yrf+0.4);
	\node[txt] at (\xant+0.05, \yrf+0.15) {$\matel{\vx}{1}$};
	\node[txt] at (\xadd+0.1-0.2, \yrf+0.4-0.2) {\scalebox{0.6}{\textbf{\tikzvdots}}};
	\draw[lsty] (aMM) -- (\xadd-0.1, -\yrf-0.4) \adder{0.15} -- (\xant, -\yrf-0.4) \antenna{0.5};
	\draw[lsty] (aM1) -| (\xadd-0.1, -\yrf-0.4);
	\node[txt] at (\xant+0.05, -\yrf-0.65) {$\matel{\vx}{M}$};
	\node[txt] at (\xadd-0.1-0.2, -\yrf-0.4+0.2) {\scalebox{0.6}{\textbf{\tikzvdots}}};
	\node[txt] (antdot) at (\xant, 0) {\tikzvdots};

	\draw[lsty] (\xuant, \yue) \antenna{0.5} -- (\xuadd, \yue) \adder{0.15};
	\node[txt] (y1) at (\xuant+1.5, \yue) {$\!y_1$};
	\draw[lsty, ->] (\xuadd, \yue) -- (y1);
	\node[txt, inner sep=0.25mm] (n1) at (\xuadd, \yue+0.6) {$n_1$};
	\draw[lsty, ->] (n1) -- (\xuadd, \yue+0.15);
	\node[txt] at (\xuant+0.25, \yue-0.35) {\scriptsize User $1$};
	\draw[lsty] (\xuant, -\yue) \antenna{0.5} -- (\xuadd, -\yue) \adder{0.15};
	\node[txt] (yI) at (\xuant+1.5, -\yue) {$\!y_I$};
	\draw[lsty, ->] (\xuadd, -\yue) -- (yI);
	\node[txt, inner sep=0.25mm] (nI) at (\xuadd, -\yue+0.6) {$n_I$};
	\draw[lsty, ->] (nI) -- (\xuadd, -\yue+0.15);
	\node[txt] at (\xuant+0.25, -\yue-0.35) {\scriptsize User $I$};
	\node[txt] (uantdot) at (\xuant, 0) {\tikzvdots};

	\draw[lsty, ->] (\xant+0.2, 0.2) -- (\xuant-0.2, \yue+0.2) node[txt, midway, sloped, above] {$\vh_1^*$};
	\draw[lsty, ->] (\xant+0.2, -0.2) -- (\xuant-0.2, -\yue+0.2) node[txt, midway, sloped, above] {$\vh_I^*$};

	\draw [lsty, decorate, decoration={brace, amplitude=2mm, mirror}] (-0.4,\ydesc) -- (0.4, \ydesc) node[txt, midway,yshift=-4.5mm]{Digital Precoder};
	\draw [lsty, decorate, decoration={brace, amplitude=2mm, mirror}] (\xrf-0.5,\ydesc) -- (\xadd+0.1, \ydesc) node[txt, midway,yshift=-4.5mm]{Analog Precoder};

\end{tikzpicture}
	\caption{Diagram of the considered hybrid system model.\label{fig:system_model}}
\end{figure}
We consider a \gls{mumiso} downlink channel with an $M$-antenna hybrid digital-analog transmitter, which is equipped with $\numrf \leq M$ \gls{rf} chains, serving $I$ single-antenna receivers as shown in Fig.~\ref{fig:system_model}.
The transmitter allocates a transmit power $p_i\geq 0$ to send symbol $\brs_i\in \Cset$ to each user $i$, where $\expect{\abs{\brs_i}^2} = 1$ and $\expect{\brs_i \brs_j^\ast} = 0$ for $i\neq j$, over a frequency-flat and block-fading channel $\vh_i^* \in \Cset^M$.
\par
For each user $i$, the transmitted symbol is formed using a hybrid beamformer $\bA \vb_i$, where the matrix $\bA = \bmat{\va_1 ~\cdots~\va_{\numrf}}$ is common to all users and consists of $\numrf$ analog beamformers $\va_r\in \cbook$ selected from a discrete finite codebook $\cbook = \uset{\bva_1,\dots,\bva_{\numcb}}$.
The vector $\vb_i \in \Cset^{\numrf}$ is the digital precoder for user $i$ and satisfies $\norm{\vb_i}_2=1$ for $i=1,\dots,I$.
Given the total transmit signal $\vx=\sum_{i=1}^I\sqrt{p_i}\bA\vb_i\brs_i$, the received signal $y_i$ at user $i$ including multi-user interference is \cite{hegdeHybridBeamformingLargescale2017}
\begin{align}
	y_i =  \sqrt{p_i}\vh_i^\He \bA\vb_i \brs_i + \sum_{j \neq i } \sqrt{p_j}\vh_i^\He \bA\vb_j \brs_j  + n_i.
\end{align}
Without loss of generality, $n_i$ is modeled as additive white Gaussian noise with unit variance $\expect[]{\abs{n_i}^2} = 1$. %
Defining $\vp = \bmat{p_1~ \cdots~ p_I}^\T$ and $\bB = \bmat{\vb_1~ \cdots~ \vb_I}^\T$ for convenience, the resulting downlink \gls{sinr} at user $i$ is characterized by
\begin{align}
	\sinrdl_i(\vp, \bB, \bA) &= \frac{p_i \vb_i^\He \bPsi_{i}(\bA) \vb_i}{\sum_{j\neq i} p_j\vb_j^\He\bPsi_{i}(\bA) \vb_j + 1}\label{eq:hybrid_sinr_dl}. %
\end{align}
Here, we define for $i=1, \dots, I$
\begin{equation}
	\bPsi_{i}(\bA) = \bA^\He \bR_i \bA,
\end{equation}
where $\bR_i\succeq \nullmat$.
In case instantaneous \gls{sinr} is considered, $\bR_i = \vh_i\vh_i^\He$ is the outer channel product.
Alternatively, in case statistical \gls{csi} is available, $\bR_i$ can be substituted by the channel covariance $\bR_i = \expect[]{\vh_i\vh_i^\He}$ such that $\bPsi_{i}(\bA)$ is the channel covariance \textit{seen} by the digital precoder, leading to $\sinrdl_i(\vp, \bB, \bA)$ being an approximation of the average \gls{sinr} \cite{shaoSimpleWayApproximate2017}.

\subsection{Energy-efficient Hybrid Beamforming with Exact CSI}
First, the case of availability of exact \gls{csi} at the transmitter is considered.
The goal is to find the power allocation $\vp^\opt$ and corresponding beamformers $\oset{\bA^\opt, \bB^\opt}$ such that the weighted sum-power $\vp^\T \weightvec(\bA, \bB)$ is minimized while a minimum \gls{sinr} $\gamma_i > 0$ at each user $i$ is guaranteed.
The weights $\weightvec(\bA, \bB) \in \Rset^I$ are given as
\begin{equation}
	\matel{\weightvec(\bA, \bB)}{i} = \vb_i^\He \bPsi_{0}(\bA) \vb_i, \label{eq:pow_weight}
\end{equation}
where $\bPsi_{0}(\bA) = \idmat$ if the baseband power shall be minimized as in \cite{hegdeHybridBeamformingLargescale2017}, or $\bPsi_{0}(\bA) = \bA^\He \bA$ if the \gls{rf} power shall be minimized as in, \eg, \cite{vazquezMultiuserDownlinkHybrid2017}.
As we will see in Sec.~\ref{sec:greedybf_perfcsi}, $\bPsi_{0}$ serves a role similar to $\bPsi_{i}$ for $i \neq 0$ in the equivalent virtual uplink system, where it represents the noise covariance.
For ease of notation, we omit the arguments of $\vw$ in the following.
Given a maximum total allocated power $\pmax$, the energy-efficient hybrid \gls{bf} problem is formulated as
\begin{equation}
	\begin{split}
		\min_{\vp, \bA, \bB} \quad& \weightvec^\T \vp\\[-1.5mm]
		\upm{s.t.} \quad & \bA \in \mathcal{A}, \bB \in \cB\\
		& \weightvec^\T \vp \leq \pmax,~~ \vp \geq 0,\\
		& (\forall i)\, \sinrdl_i(\vp, \bB, \bA) \geq \gamma_i,%
	\end{split}\label{eq:problem_dl_exactcsi_hybrid}
\end{equation}
for which we define the set of possible analog beamformer matrices $\bA$ and baseband \gls{bf} matrices $\bB$, respectively, as
\begin{alignat}{2}
	\mathcal{A} &= \big\{ \bmat{\va_1~\cdots~\va_{\numrf}}  &&\big| (\forall r \in \uset{1,\dots,\numrf} )~ \va_r\in \cbook, \nonumber\\ 
	&&&~\rank{\bmat{\va_1~\cdots~\va_{\numrf}}} = \numrf \big\}, \label{eq:analogbf_set}\\
	\cB &= \big\{ \bmat{\vb_1 ~\cdots~\vb_I}  &&\big| (\forall i \in \uset{1,\dots,I}) \nonumber\\
	&&&~ \vb_i \in\Cset^{\numrf} \land \norm{\vb_i}_2 = 1\big\}.
\end{alignat}
Note that in \eqref{eq:analogbf_set} we restrict $\bA$ to a selection of linearly independent codewords.
\subsection{Energy-efficient Hybrid Beamforming with Inexact CSI}
Next, we assume that only inexact \gls{csi}
\begin{align}
	\hbR_i = \bR_i + \tbR_i, \label{eq:imperfect_R}
\end{align}
is available at the transmitter, where $\tbR_i \in \Cset^{M\times M}$ such that $\hbR \succeq \nullmat$ is the \gls{csi} error resulting from, \eg, channel estimation error or feedback quantization.
Furthermore, let $\prdf{\cS}$ denote a probability distribution of system realizations $\cS = \oset{\oset{\bR_i, \tbR_i, \gamma_i, \vxi_i}_{i=1}^I, \pmax}$, where $\vxi_i$ denotes available auxiliary features that encode information about the characteristics of the \gls{csi} estimate.
For instance, $\vxi_i$ may represent the quantization level of the channel feedback or the pilot power used during channel sounding.
Specific choices of $\vxi_i$ are discussed in Sec.~\ref{sec:results}.
In addition, we define a mapping
\begin{equation}
	\cM(\cS) = \oset*{\vp_{\cM}(\cS), \bA_{\cM}(\cS), \bB_{\cM}(\cS)}
\end{equation}
where $\vp_{\cM}(\cS)$, $\bA_{\cM}(\cS)$ and $\bB_{\cM}(\cS)$ are individual mappings from realizations $\ssample$ onto the power allocation $\vp$, the analog beamformers $\bA$ and digital precoders $\bB$, respectively.
Under the scenario distribution characterized by $\prdf{\cS}$, our aim is to find a mapping $\cM(\cdot)$ that, without accessing the ground-truth \gls{csi} $\bR_i$, minimizes the expected allocated power while the \gls{qos} constraints $\sinrdl_i\pdel*{ \cM\oset{\cS}}\geq \gamma_i$ are satisfied with probability $1 - \pout$, where $\pout$ is a nominal outage probability.
Given a mapping $\cM(\cdot)$, satisfaction of this probabilistic \gls{qos} constraint is equivalent to non-positivity of
\begin{align}
	&g_i(\cM, \prdf{\ssample}) \nonumber \\
	&= 1-\pout - \prob[\ssample \sim \prdf{\cS}]{\sinrdl_i\pdel*{ \cM\oset{\cS}}\geq \gamma_i}, \label{eq:prob_qos_constraint}
\end{align}
where the rightmost term is the non-outage probability.
Note that both $\sinrdl_{i}$ and $\gamma_i$ implicitly depend on the realization $\cS$.
The desired robust energy-efficient hybrid \gls{bf} problem, therefore, can be formulated as
\begin{equation}
	\begin{alignedat}{2}
		\min_{\cM} & \quad \expect[\cS\sim \prdf{\cS}]{ \weightvec^\T \vp_\cM({\cS})}\\
		\quad \upm{s.t.} & \quad\!  \big(\forall\cS \in \supp(\prdf{\cS}) \big) \begin{cases}
			\bA_{\cM}({\cS}) \in \mathcal{A}\\
			\bB_{\cM} ({\cS}) \in \cB\\
			\vp_\cM({\cS}) \geq \nullvec \\
			\weightvec^\T\vp_\cM({\cS}) \leq P_\upm{max}
		\end{cases}\\
		&\quad (\forall i)~ g_i(\cM, \prdf{\cS}) \leq 0.
	\end{alignedat}\label{eq:problem_dl_inexactcsi}
\end{equation}
Unlike \gls{qos} constraints based on \textit{expected} values \cite{eisenOptimalWirelessResource2020} or worst-case robust designs \cite{wajidRobustDownlinkBeamforming2013}, outage constraints prevent extreme cases dominating the solution, thus ensuring that the majority of users achieve the target \gls{qos} level $\gamma_i$.

\section{Hybrid Downlink Beamforming and Power Allocation for Perfect CSI}\label{sec:greedybf_perfcsi}\label{sec:hybridperfectcsi_alg}

Before considering robust energy-efficient hybrid \gls{bf} using imperfect \gls{csi}, we first return to the exact \gls{csi} problem \eqref{eq:problem_dl_exactcsi_hybrid}.
In \cite{hegdeHybridBeamformingLargescale2017}, an optimal and two sub-optimal algorithms are proposed to solve the optimization problem.
While \cite{hegdeHybridBeamformingLargescale2017} only considers the case $\weightvec=\onevec$, the generalization to \eqref{eq:pow_weight} is straightforward.
\par
The optimal algorithm evaluates the reduced $\numrf$-port digital \gls{bf} subproblem with fixed $\bA$
\begin{equation}
	\begin{split}
		\min_{\vp, \bB} \quad& \weightvec^\T \vp\\[-1.5mm]
		\upm{s.t.} \quad & \bB \in \cB,\\
		& \weightvec^\T \vp \leq \pmax,~~ \vp \geq 0,\\
		& (\forall i)~ \sinrdl_i(\vp, \bB, \bA) \geq \gamma_i,%
	\end{split}\label{eq:problem_dl_exactcsi_digital}
\end{equation}
for each configuration $\bA \in \mathcal{A}$ (excluding equivalent configurations under permutation). %
While optimal, the growth in the number of configurations renders the optimal algorithm computationally impractical.
Instead, we focus on the sub-optimal greedy correction algorithm in \cite[Alg.~2]{hegdeHybridBeamformingLargescale2017}.
Define %
\begin{equation}
	\oset*{\vp^\opt(\bPsi(\bA)), \bB^\opt(\bPsi(\bA))}=\digbfmap\left(\bPsi(\bA)\right)
\end{equation}
with $\bPsi(\bA) = \oset*{\bPsi_{i}(\bA)}_{i}$ as the mapping onto the solution of the digital \gls{bf} subproblem in \eqref{eq:problem_dl_exactcsi_digital} for some choice of $\bA$.
In each iteration $\ell=1, \dots, \numaly$ of the greedy correction algorithm, the idea is to select a particular analog \gls{rf} chain $r\idx{\ell}\in \uset{1,\cdots, \numrf}$ to update, then obtain $\bA\idx{\ell}$ as the correction
\begin{align}
	\bA_{r\idx{\ell}}\idx{\ell-1}\pdel*{\!\va_{r\idx{\ell}}\idx{\ell}\!}
	\!=\! \bmat*{\va_1\idx{\ell-1}\cdots\va_{r\idx{\ell}-1}\idx{\ell-1}~\va_{r\idx{\ell}}\idx{\ell} ~ \va_{r\idx{\ell}+1}\idx{\ell-1}\cdots \va_{\numrf}\idx{\ell-1}},
\end{align}
where the updated beam $\va_{r\idx{\ell}}\idx{\ell}$ is chosen such that it minimizes the total allocated power, \ie, $\va_{r\idx{\ell}}\idx{\ell} = \bva_{a\idx{\ell}}$ with
\begin{equation}
	a\idx{\ell} = \argmin_{\substack{a \in \uset{1,\dots,\numcb}}}  \weightvec^\T\vp^\opt \pdel*{\bPsi\pdel*{\bA\idx{\ell-1}_{r\idx{\ell}}(\bva_a)}}.
\end{equation}
Since $ \weightvec^\T\vp^\opt \pdel*{\bPsi\pdel*{\bA\idx{\ell}}} \leq \weightvec^\T\vp^\opt \pdel*{\bPsi\pdel*{\bA\idx{\ell-1}}}$ for any sequence of $\oset{r\idx{\ell}}_\ell$ with feasible initialization, the algorithm converges.
\par
The \gls{bf} subproblem in \eqref{eq:problem_dl_exactcsi_digital} is established, and several algorithms are proposed in the literature to find a feasible solution \cite{songNetworkDualityMultiuser2007,rashid-farrokhiJointOptimalPower1998, schubertSolutionMultiuserDownlink2004}.
Commonly, these algorithms are based on the dual virtual uplink problem \cite{songNetworkDualityMultiuser2007}
\begin{equation}
	\begin{split}
		\min_{\vq, \bB} \quad& \onevec^\T\vq\\[-1.5mm] %
		\upm{s.t.} \quad & \bB \in \cB\\
		& \onevec^\T\vq \leq \pmax,~~ \vq \geq 0,\\
		& (\forall i)~ \sinrul_i(\vq, \bB, \bA) \geq \gamma_i,\\
	\end{split}\label{eq:problem_ul_exactcsi_digital}
\end{equation}
where $\vq = \bmat{q_1 ~~ \cdots~~ q_I}^\T$ is the uplink power allocation and 
\begin{equation}
	\sinrul_i(\vq, \bB, \bA) = \frac{q_i \vb_i^\He \bPsi_{i}(\bA) \vb_i}{\sum_{j\neq i} q_j \vb_i^\He\bPsi_{j}(\bA)\vb_i + \vb_i^\He\bPsi_{0}(\bA) \vb_i} \label{eq:hybrid_sinr_ul}
\end{equation}
is the virtual uplink \gls{sinr}, to solve the problem efficiently.
It is shown in, \eg, \cite{songNetworkDualityMultiuser2007} or by slightly modifying the derivations in \cite{schubertSolutionMultiuserDownlink2004}, that problem \eqref{eq:problem_dl_exactcsi_digital} is feasible if and only if problem \eqref{eq:problem_ul_exactcsi_digital} is feasible. 
Further, if optimal points exist, then the weighted downlink power and virtual uplink power at the optimum are equal ($\weightvec^\T\vp^\opt = \onevec^\T\vq^\opt$) and the optimal beamformers $\bB^\opt$ are identical.
Consequently, it suffices to leverage the solution of \eqref{eq:problem_ul_exactcsi_digital}
\begin{equation}
	\oset*{\vq^\opt(\bPsi(\bA)), \bB^\opt(\bPsi(\bA))}=\digbfmapul\left(\bPsi(\bA)\right)
\end{equation}
and update $\bA$ and $\bB$ according to\footnote{If $\bA\idx{\ell-1}_{r\idx{\ell}}(\bva_a) \notin \cA$, we declare $\weightvec^\T\vq^\opt \pdel*{\bPsi\pdel*{\bA\idx{\ell-1}_{r\idx{\ell}}(\bva_a)}} = \infty$.}
\begin{align}
	\bA\idx{\ell} &= \bA\idx{\ell-1}_{r\idx{\ell}}(\bva_{a\idx{\ell}}), \label{eq:analog_bf_selection_ul}\\\allowdisplaybreaks[1]
	\bB\idx{\ell} &= \bB^\opt \pdel*{\bPsi\pdel*{\bA\idx{\ell-1}_{r\idx{\ell}}(\bva_{a\idx{\ell}})}}, \\\allowdisplaybreaks[1]
	\text{with} \quad
	a\idx{\ell} &= \argmin_{\substack{a \in \uset{1,\dots,\numcb}}} %
	\onevec^\T\vq^\opt \pdel*{\bPsi\pdel*{\bA\idx{\ell-1}_{r\idx{\ell}}(\bva_a)}} . 
\end{align}
Finally, utilizing $\bB\idx{\numaly}$, the downlink power allocation $\vp^\star$ can be easily recovered by solving
\begin{align}
	\bC_{\upm{C}} (\bB\idx{\numaly}, \bA\idx{\numaly})\vp^\opt = \onevec,
\end{align}
where the coupling relationship is described by the matrix $\bC_\upm{C}(\bB, \bA) \in \Rset^{I\times I}$ with \cite{schubertSolutionMultiuserDownlink2004}
\begin{align}
	[\bC_\upm{C}(\bB, \bA)]_{i,j}& =\begin{cases}
		\gamma_i^{-1}\vb_i^\He \bPsi_{i}(\bA)\vb_i & \text{~~for~~$i=j$.}\\
		-\vb_j^\He \bPsi_{i}(\bA) \vb_j & \text{~~for~~$i\neq j$,}\\
	\end{cases} \label{eq:couplingI_matrix}
\end{align}

\par
In this work, we solve \eqref{eq:problem_ul_exactcsi_digital} by a variation of \cite[Alg.~2]{schubertSolutionMultiuserDownlink2004} that replaces the unity uplink noise variance by the term $\vb_i^\He\bPsi_{0}(\bA) \vb_i$ resulting from weighting the downlink power.

\section{Robust Unrolled Hybrid Downlink Beamforming and Power Allocation}\label{sec:prop_robust_greedybf}
\subsection{Proposed Model Architecture}\label{sec:hybrid_architecture}
\begin{figure*}[ht]
	\centering
	\newcommand\gnncoordy{0.75}

\newcommand\fax{0}
\newcommand\fay{-1.2}

\tikzset{external/export next=false} %
\begin{tikzpicture}[]
	
	\node[boxgray] (m1) at (0.6, 0) {$\cF_{\upm{A}}$};
	\node[boxgray, fill=lightBlue] (m2) at (2.5, 0) {$\cF_{\upm{A}}$};
	\node[txt] (mdot) at (4.0, 0) {$\cdots$};
	\node[boxgray] (m3) at (6.0, 0) {$\cF_{\upm{A}}$};
	\node[boxgray] (m4) at (8.0, 0) {\eqref{eq:dlpow_calc_apx}};
	\node[boxgray] (m5) at (10.0, 0) {\eqref{eq:pow_projection}};
	\node[txt] (i0) at (-1.4, 0) {$\bA\idx{0}$};
	\node[txt] (iL) at (11.5, 0) {$\bA_{\cM}$\\[-1mm]$\bB_{\cM}$\\[-1mm]$\vp_{\cM}$};
	
	\draw[->, lsty] ([yshift=0mm]i0) -- ([yshift=0mm]m1.west);
	\draw[->, lsty] ([yshift=0mm]m1.east) -- ([yshift=0mm]m2.west) node[midway, above, txt] {$\bA\idx{1}$};
	\draw[-, lsty] ([yshift=0mm]m2.east) -- ([yshift=0mm]mdot.west) node[midway, above, txt] {$\bA\idx{2}$};
	\draw[->, lsty] ([yshift=0mm]mdot.east) -- ([yshift=0mm]m3.west) node[midway, above, txt] {$\bA\idx{\numaly-1}$};
	\draw[->, lsty] ([yshift=0mm]m3.east) -- ([yshift=0mm]m4.west) node[midway, above, txt] {$\bA\idx{\numaly}$\\[-1mm]$\bB\idx{\numaly}$};
	\draw[->, lsty] ([yshift=0mm]m4.east) -- ([yshift=0mm]m5.west) node[midway, above, txt] {$\bA\idx{\numaly}$\\[-1mm]$\bB\idx{\numaly}$\\[-1mm]$\vp^\opt$};
	\draw[->, lsty] ([yshift=0mm]m5.east) -- ([yshift=0mm]iL.west);
	
	\node[boxgray, fill=lightRed] (gnn) at (-1.75, \gnncoordy) {$\gcnn(\cdot, \vtheta)\!$};
	\node[txt] (gnnin) [left=0.5 of gnn] {$~\bZ\idx{0}$\\[1mm]$\bS$};
	\node[txt] (gnndots) [above=\gnncoordy of mdot.center, anchor=center] {$\cdots$};
	\node[bdot] (gnnd1) [above=\gnncoordy of m1.center, anchor=center] {};
	\node[bdot] (gnnd2) [above=\gnncoordy of m2.center, anchor=center] {};
	\node[bdot] (gnnd3) [above=\gnncoordy of m3.center, anchor=center] {};
	
	\draw[->, lsty] ([yshift=0mm]gnnin.east) -- (gnn.west);
	\draw[->, lsty] ([yshift=0mm]gnn.east) -| (m1.north) node[near start, above, txt] {$\bZ\idx{\numgnnly}$};
	\draw[->, lsty] ([yshift=0mm]gnn.east) -| (m2.north);
	\draw[-, lsty] ([yshift=0mm]gnn.east) -- (gnndots.west);
	\draw[->, lsty] ([yshift=0mm]gnndots.east) -| (m3.north);
	\draw[->, lsty] ([yshift=0mm]gnndots.east) -| (m4.north);

	\node[boxgray] (r1) at (\fax, \fay) {$\bA_{r\idx{2}}(\bva_1)$};
	\node[txt] (rdot) [below=0.7 of r1.center, anchor=center] {\tikzvdots};
	\node[boxgray] (rB) [below=0.7 of rdot.center, anchor=center] {$\bA_{r\idx{2}}(\bva_\numcb)$};
	
	\node[boxgray] (chest1) [right=2.5 of r1.center, anchor=center] {$\eqref{eq:est_ch_model}$};
	\node[txt] (chestdot) [right=2.5 of rdot.center, anchor=center] {\tikzvdots};
	\node[boxgray] (chestB) [right=2.5 of rB.center, anchor=center] {$\eqref{eq:est_ch_model}$};
	
	\draw[->, lsty] ([yshift=0mm]r1.east) -- (chest1.west);
	\draw[->, lsty] ([yshift=0mm]rB.east) -- (chestB.west);
	
	\node[boxgray] (fd1) [right=2 of chest1.center, anchor=center] {$\digbfmapul$};
	\node[txt] (fddot) [right=2 of chestdot.center, anchor=center] {\tikzvdots};
	\node[boxgray] (fdB) [right=2 of chestB.center, anchor=center] {$\digbfmapul$};
	
	\draw[->, lsty] ([yshift=0mm]chest1.east) -- (fd1.west) node[midway, above, txt] {$\hbPsi$};
	\draw[->, lsty] ([yshift=0mm]chestB.east) -- (fdB.west) node[midway, above, txt] {$\hbPsi$};
	
	\node[boxgray, minimum height=2.2cm] (min) [right=3.5 of fddot.center, anchor=center] {Select\\$\argmin_a$\\$ \norm{\vq\idx{2,a}}_1$};
	\draw[->, lsty] ([yshift=0mm]fd1.east) -- ([yshift=0.7cm]min.west) node[txt, midway, above] {$\vq\idx{2, 1}$,$\bB\idx{2, 1}$};
	\draw[->, lsty] ([yshift=0mm]fdB.east) -- ([yshift=-0.7cm]min.west) node[txt, midway, above] {$\vq\idx{2, \numcb}$,$\bB\idx{2, \numcb}$};
	
	\node[txt] (fain) [left=2.0 of r1.center] {$\bA\idx{1}$};
	\node[bdot] (fainbdot) [left=1.3 of r1.center, anchor=center] {};
	\node[txt] (faindots) [left=1.3 of rdot.center, anchor=center] {\tikzvdots};
	
	\draw[->, lsty] ([yshift=0mm]fain.east) -- (fainbdot) -- (r1.west);
	\draw[-, lsty] (fainbdot) -- (faindots.north);
	\draw[->, lsty] (faindots) |- (rB.west);
	
	\node[txt] (faout) [right=1.5 of min.center] {$\bA\idx{2}$\\[-1mm]$\bB\idx{2}$};
	\draw[->, lsty] (min.east) -- (faout);
	
	\node[txt] (zin) [above=0.3 of fain.center] {$\bZ\idx{\numgnnly}$};
	\node[bdot] (zbdot) [above left=0.15 and 0.8 of chest1.center, anchor=center] {};
	\node[txt] (zdots) [left=0.8 of chestdot.center, anchor=center] {\tikzvdots};
	\draw[-_, lsty] (zin) -| (zbdot.center);
	\draw[->, lsty] (zbdot) -- ([yshift=0.15cm]chest1.west);
	\draw[-, lsty] (zbdot) -- (zdots);
	\draw[->, lsty] (zdots) |- ([yshift=0.15cm]chestB.west);
	
	\node[txt] (ziny) [above=0.3 of r1.center] {}; %
	
	\begin{scope}[on background layer]
		\node[nl, fit=(fainbdot)(r1)(rB)(min)(ziny)] {};
	\end{scope}

\end{tikzpicture}
	\caption{Block diagram of the proposed \gls{dn} architecture $\cM(\cS; \vtheta)$ for outage-constrained hybrid \gls{bf}. $\cF_{\upm{A}}$ represent one analog beamformer selection step, while $\digbfmapul$ represents solving the virtual uplink problem \eqref{eq:problem_ul_inexactcsi_digital}. \label{fig:hybrid_unrolled}}
\end{figure*}
We return to the case of imperfect \gls{csi}.
Classical robust \gls{bf} methods typically handle \gls{csi} uncertainty by explicit substitution of $\bR_i$ with $\hbR_i - \tbR_i$ in the \gls{qos} model. %
For the downlink \gls{sinr}, for instance, this results in
\begin{align}
	&\sinrdlest_i\pdel{\vp, \bB, \bA} \nonumber\\
	&\quad= \frac{p_i \pdel{\vb_i^\He \bA^\He\hbR_i\bA\vb_i -\vb_i^\He \bA^\He\tbR_i\bA\vb_i} }{\sum_{j\neq i} p_j \pdel{\vb_j^\He \bA^\He\hbR_i\bA\vb_j - \vb_j^\He \bA^\He\tbR_i\bA\vb_j} + 1}. \label{eq:hybrid_sinr_dl_error}
\end{align}
If the probability distribution of the error is assumed to be known and tractable, \eg, Gaussian or uniform, the probabilistic constraints can be translated into second-order cone or semidefinite constraints by finding suitable bounds of the error terms involving $\tbR_i$ (\cf \cite{wangOutageConstrainedRobust2014, botrosshenoudaProbabilisticallyConstrainedApproachesDesign2008}).
Albeit these approaches control outage jointly across multiple users, the error distribution is typically not known in practice and the respective convex optimization problems incur a high computational cost.
Alternatively, a computationally efficient approach for introducing robustness is to utilize a deterministic approximation $\bE_i \approx \tbR_i$ of the error, \eg, diagonal loading \cite{ bjornsonOptimalMultiuserTransmit2014, nguyenMMSEPrecodingMultiuser2014, bjornsonMassiveMIMOImperfect2016}.
Effectively, a \textit{virtual} channel is considered.
However, finding approximations that realize particular outage levels in \gls{mu} systems and, at the same time, are efficient \wrt the allocated power is not straightforward.
\par
We therefore propose a model-aided \gls{dn} architecture $\cM(\cS; \vtheta) = (\vp_{\cM}(\cS; \vtheta), \bA_{\cM}(\cS; \vtheta), \bB_{\cM}(\cS; \vtheta))$ with parameters $\vtheta$ based on the greedy correction algorithm \cite[Alg.~2]{hegdeHybridBeamformingLargescale2017} that is reviewed in Sec.~\ref{sec:hybridperfectcsi_alg}.
Inspired by the deterministic approximation approach, it leverages a dynamic \gls{gnn}-based representation of $\bE_i$ that adapts to scenario instances, thus implicitly learning the error distribution from data.
\par
Specifically, we consider the \gls{qos} model
\begin{align}
	\sinrdlest_i(\vp, \bB, \bA) &= \frac{p_i \vb_i^\He \hbPsi_{i,i}(\bA)\vb_i}{\sum_{j\neq i} p_j\vb_j^\He \hbPsi_{i,j}(\bA)\vb_j + 1},\\
	\sinrulest_i(\vq, \bB, \bA) &= \frac{q_i \vb_i^\He \hbPsi_{i,i}(\bA)\vb_i}{\sum_{j\neq i} q_j\vb_i^\He \hbPsi_{j,i}(\bA)\vb_i + \vb_i^\He\hbPsi_{0, i}(\bA) \vb_i} \label{eq:mod_hybrid_ulsinr},
\end{align}
where positive semidefinite virtual channel covariances $\hbPsi_{i,j}$ for $i = 1, \dots, I$ and $j = 1, \dots, I$ are given as
\begin{align}
	&\hbPsi_{i,j}(\bA) \!=\! \begin{cases}
		\! \bA^\He \!\left(z_{i,1} \hbR_i  + z_{i,2} \frac{\trace{\hbR_i}}{M}  \idmat\right)\!\bA& \!\!\!\text{for $i=j$,}\\
		\! \bA^\He \!\left(z_{i,3} \hbR_i + z_{i,4} \frac{\trace{\hbR_i}}{M}  \idmat\right)\!\bA & \!\!\!\text{for $i\neq j$.}
	\end{cases}\label{eq:est_ch_model}
\end{align}
For $i=0$ we simply choose the original uplink noise covariance $\hbPsi_{0,i}(\bA) = \bPsi_{0}(\bA)$.
The coefficients $z_{i,1} > 0$ and  $z_{i,3} > 0$ model error components proportional to the quadratic forms $\vb_j^\He \bA^\He \hbR_i \bA \vb_j$ in \eqref{eq:hybrid_sinr_dl_error}, \ie, the correlation between beamformers and channel estimates. 
The remaining coefficients $z_{i,2}>0$ and $z_{i,4}>0$ model error components that are independent of the spatial structure of the channel estimate and are only proportional to the gain $\trace{\hbR_i}$.
Note that we allow separate error approximations for the wanted and interfering signals to account for the dependence on the beamformers $\vb_j$ of interfering users.
\par
The coefficients matrix $\bZ_\upm{out} \in \Rset^{I\times F_\upm{out}}$ with $F_\upm{out} = 4$ and entries $\matel{\bZ_\upm{out}}{i, f} = z_{i,f}$ representing the error components in \eqref{eq:est_ch_model} are estimated in an instance-adaptive manner by a \gls{gcnn} $\gcnn(\,\cdot\,;\vtheta): \Rset^{I\times F\idx{0}} \times \Rset^{I \times I} \to \Rset ^{I\times F\idx{\numgnnly}}$ \cite{kipfSemiSupervisedClassificationGraph2017} with $\numgnnly$ layers, where $F\idx{0} = F_\upm{in}$ is the number of input features and $F\idx{\numgnnly} = F_\upm{out}$.
The transformation of latent features by the \gls{gcnn} $\gcnn = \gcnnly\idx{\numgnnly} \circ \cdots \circ \gcnnly\idx{1}$ is characterized by the nonlinear map
\begin{align}
	\bZ\idx{\ell'} &= \gcnnly\idx{\ell'}\big(\bZ\idx{\ell'-1}, \bS; \vtheta\idx{\ell'}\big) \nonumber\\
	&= \phi\idx{\ell'} \Bigg(\sum_{f=0}^{\numgnnfilt} \bS^f \bZ\idx{\ell'-1} \bTheta_f\idx{\ell'} + \onevec \left(\vtheta_{\upm{b}}\idx{\ell'}\right)^\T\Bigg),
\end{align}
where $\ell'$ is the \gls{gcnn} layer index, $\phi\idx{\ell'}$ is an elementwise nonlinearity, $\vtheta\idx{\ell'} = (\bTheta_0\idx{\ell'}, \dots, \bTheta_{\numgnnfilt}\idx{\ell'}, \vtheta_{\upm{b}}\idx{\ell'})$ are learnable weights with $\bTheta_f\idx{\ell'} \in \Rset^{F\idx{\ell'-1} \times F\idx{\ell'}}$ and $\vtheta_{\upm{b}}\idx{\ell'} \in \Rset^{\idx{\ell'}}$, and $\numgnnfilt$ is the filter degree.
A \gls{gcnn}, which is a type of \gls{gnn}, is a \gls{dn} architecture tailored for graph-structured data.
Each row $\matel{\bZ\idx{\ell'}}{i, :}$ can be interpreted as the (latent) feature vector of a node of a graph, in this case, the users of the system.
Here, the shift operator $\bS$ encodes the associated graph topology, specifically, the estimated correlation coefficient between channels of users
\begin{align}
	\matel{\bS}{i_1,i_2} = \begin{cases}
		\frac{\abs*{\trace{\hbR_{i_1} \hbR_{i_2}}}}{\fnorm*{\hbR_{i_1}}\fnorm*{\hbR_{i_2}}} & \text{for~}i_1\neq i_2,\\
		0 & \text{for~}i_1=i_2.
	\end{cases}
\end{align}
We define the input feature vector per user as
\begin{equation}
	\matel{\bZ\idx{0}}{i, :} = \bmat*{\log\fnorm{\hbR_i}^2~~\log\gamma_i~~\log\vxi_i^\T~~\log\pmax}^\T. \label{eq:in_features}
\end{equation}
\par
On the one hand, the topology defined by the shift operator enables the \gls{gcnn} to predict the interference-related coefficients $z_{i,3}$ and $z_{i,4}$ associated with pairs of users.
Compared to alternative \gls{dn} architectures, \glspl{gnn}, and thus \glspl{gcnn} by extension, have the advantage that they can be applied to similar graph-structured data with any number of nodes (users) given the same set of learnable weights \cite{battagliaRelationalInductiveBiases2018}.
Moreover, the mapping provided by \glspl{gnn} is \textit{permutation equivariant}, guaranteeing that if the input graph is permuted, \ie, the users are relabeled, the output does not change except for the same permutation \cite{gamaGraphsConvolutionsNeural2020}.
This is a necessary property for resource allocation communication systems, where we expect that permutations of users or antennas should lead to corresponding permutations of the optimal resources.
On the other hand, the input feature vector enables adaptation to the realization of the channel state, the \gls{csi} quality and the \gls{qos} requirements.
The logarithmic mapping compresses the value range and facilitates generalization across different scales of input features.
\par
The integration of the \gls{gcnn} into the greedy correction algorithm is illustrated in Fig.~\ref{fig:hybrid_unrolled}, thereby denoting one iteration of the analog beamformer selection as $\cF_{\upm{A}}$.
In $\digbfmapul$, instead of solving Problem~\eqref{eq:problem_ul_exactcsi_digital}, we substitute $\sinrulest_i$ into $\sinrul_i$ and omit the power constraint $\onevec^\T \vq \leq \pmax$.
In particular, we solve
\begin{equation}
	\begin{alignedat}{2}
		\min_{\vq \geq \nullvec, \bB\in \cB} \quad& \onevec^\T \vq \\[-1.5mm]
		\upm{s.t.} \quad & (\forall i)~ \sinrulest_i(\vq, \bB, \bA) \geq \gamma_i,%
	\end{alignedat}\label{eq:problem_ul_inexactcsi_digital}
\end{equation}
where $\sinrulest_i$ is defined in \eqref{eq:mod_hybrid_ulsinr}.
The downlink power allocation $\vp^\opt$ is recovered by solving
\begin{equation}
	\hbC_\upm{C}(\bB\idx{\numaly}, \bA\idx{\numaly})\vp^\opt = \onevec, \label{eq:dlpow_calc_apx}
\end{equation}
where $\hbC_\upm{C}(\bB\idx{\numaly}, \bA\idx{\numaly})$ is \eqref{eq:couplingI_matrix} with $\hbPsi_{i,j}$ accordingly substituted for $\bPsi_{i}$, \ie,
\begin{align}
	[\hbC_\upm{C}(\bB, \bA)]_{i,j}& =\begin{cases}
		\gamma_i^{-1}\vb_i^\He \hbPsi_{i, i}(\bA)\vb_i & \text{~~for~~$i=j$.}\\
		-\vb_j^\He \hbPsi_{i, j}(\bA) \vb_j & \text{~~for~~$i\neq j$.}\\
	\end{cases}\label{eq:apx_coupling_matrix}
\end{align}
Note that since $\hbPsi$ is constant across iterations, as in \cite[Alg.~2]{hegdeHybridBeamformingLargescale2017}, the total allocated power is a decreasing sequence over iterations $\ell$ if the initial point is feasible.
Finally, to guarantee feasibility \wrt the power constraints, $\vp^\opt$ is projected onto the feasible set by%
\begin{align}
	\vp_\cM(\ssample) = \frac{\pmax}{\pmax + \bdel*{\weightvec^\T[\vp^\opt]_0^\infty - \pmax}_0^\infty }[\vp^\opt]_0^\infty, \label{eq:pow_projection}
\end{align}
where the notation $[\vx]_0^\infty$ indicates an elementwise projection onto the interval $[0, \infty)$.

\subsection{Initialization}
A \textit{good} initialization of the analog configuration $\bA\idx{0}$ is essential for the existence of a solution to the optimization in \eqref{eq:problem_ul_inexactcsi_digital}.
Otherwise, the subproblems to select codewords $\va_{r\idx{\ell}}=\bva$ may be infeasible for any $\bva \in \cbook$.
In this work, we utilize the heuristic
\begin{align}
	\left(\va_1\idx{0}, \dots, \va_{\numrf}\idx{0}\right) = {}^{\numrf} \!\argmax_{\bva \in \cbook} \sum_{i=1}^I \frac{\bva^\He \hbR_i\bva}{\trace{\hbR_i}} %
\end{align}
that effectively chooses the $\numrf$ beams that maximize the SNR \wrt the sum of normalized channels.

\subsection{Model Gradient} \label{eq:digbfmap_grad}
End-to-end learning of the parameters $\vtheta$ of the proposed model $\cM(\cS; \vtheta)$ described in Sec.~\ref{sec:prop_robust_greedybf}-\ref{sec:hybrid_architecture} requires the Jacobian of the operators that compose the model, see Fig.~\ref{fig:hybrid_unrolled}.
There are two significant obstacles.
First, the beam selection in \eqref{eq:analog_bf_selection_ul} is nondifferentiable.
Secondly, the digital beamformers $\bB^\opt$ and power allocation $\vp^\opt$ result from nested optimization problems.
Since these typically do not admit closed-form solutions, the gradient \wrt $\hbPsi(\bA)$ is not easily computable.

\subsubsection{Beam Selection}
An approximate partial derivative of the selected beam $\va_{r\idx{\ell}}\idx{\ell}$ \wrt the trial uplink power allocations $\vq\idx{\ell, a} = \vq^\opt\pdel*{\hbPsi\pdel*{\bA_{r\idx{\ell}}\idx{\ell-1}(\bva_{a})}}$ can be obtained by smoothing \eqref{eq:analog_bf_selection_ul} into the Softmin-selection \cite{schynolCodebookBasedDownlinkBeamforming2024, agustssonSofttoHardVectorQuantization2017}
\begin{align}
	\va_{r\idx{\ell}}\idx{\ell} =  \frac{\sum_{a=1}^{\numcb}
		\bva_{a} \eul^{-\beta_\upm{M} \norm{\vq\idx{\ell, a}}_1  / q_\upm{min}\idx{\ell}}
	}{
		\sum_{a=1}^{\numcb}\eul^{-\beta_\upm{M} \norm{\vq\idx{\ell, a}}_1 /q_\upm{min}\idx{\ell} }
	},
	\label{eq:beam_selection_apx}
\end{align}
where $\beta_\upm{M}>0$ is an annealing parameter controlling the trade-off between smoothness and bias of the gradient approximation.
The softmin operator approaches the $\argmin$ selection as $\beta_\upm{M} \to \infty$. 
Since the softmin operator is not scale-invariant, we propose an instance-adaptive normalization $q_\upm{min}\idx{\ell} = \min_a \norm{\vq\idx{\ell, a}}_1$ to decouple the difference of exponents belonging to different analog \gls{bf} codewords $\bva_a$ in $\eqref{eq:beam_selection_apx}$ from the absolute scale of the allocated power. %
\par
A disadvantage of smoothing between analog beams, apart from primal infeasibility, is the increased correlation between the combined channels $\hbR_i^{1/2}\va_r\idx{\ell}$.
The reduced spatial separability between users leads to ill-conditioned or infeasible \gls{bf} problems. %
To avoid this undesired effect, we instead utilize a deterministic variant of the Straight-Through Gumbel-Softmax approach \cite{jangCategoricalReparameterizationGumbelSoftmax2017} during training.
In the forward pass, analog beamformers $\va_{r\idx{\ell}}\idx{\ell}$ are selected according to the $\argmin$ selection in \eqref{eq:analog_bf_selection_ul}.
When computing the gradient in the backward pass, however, the partial derivatives $\diff\va_{r\idx{\ell}}\idx{\ell} / \diff\vq\idx{\ell, a}$ are based on the differentiable proxy in \eqref{eq:beam_selection_apx}.
\subsubsection{Digital Uplink Beamforming Problem}
Next, we consider the gradient of the digital uplink \gls{bf} problem in \eqref{eq:problem_ul_inexactcsi_digital}, which is represented by $\digbfmapul$ in Fig.~\ref{fig:hybrid_unrolled}.
Specifically, we discuss the Jacobi matrix of the optimal point $(\vq\idx{\ell,a}, \bB\idx{\ell,a})$ \wrt the arguments $\uhvpsi = \oset{\uhvpsi_{i,j}}_{i,j}$ with
\begin{equation}
	\uhvpsi_{i,j} = \bmat*{\vectd{\hbPsi_{i,j}}^\T ~~ \vectr{\hbPsi_{i,j}}^\T ~~\vecti{\hbPsi_{i,j}}^\T}^\T,
\end{equation}
where the subscripts $\upm{d}$, $\upm{re}$ and $\upm{im}$ indicate the diagonal, off-diagonal real and imaginary elements, respectively, thereby accounting for the Hermitian structure of virtual channel covariance $\hbPsi_{i,j}$.
Efficient algorithms for solving Problem~\eqref{eq:problem_ul_inexactcsi_digital} in the case of general \gls{psd} channels $\hbPsi_{i,j}$ rely on cascaded eigendecompositions \cite{schubertSolutionMultiuserDownlink2004}.
Since the derivative of the eigenvectors is undefined in case of repeated eigenvalues \cite{magnusDifferentiatingEigenvaluesEigenvectors1985}, a straightforward truncation and subsequent backpropagation through the underlying algorithm leads to numerical instability.
\par
We therefore seek a method for obtaining the gradient that is independent of the underlying solver. 
Given a parameterized convex optimization problem with differentiable objective and constraint functions, under suitable regularity conditions (\ie, if strong duality holds), the optimal points are identifiable as the solution to the \gls{kkt} conditions \cite{boydConvexOptimization2004}.
By summarizing the first-order stationarity condition, potential equality conditions and complementary slackness conditions of the \gls{kkt} conditions as a real-valued \gls{snle} $\vg(\bvzeta, \uhvpsi)=\nullvec$, where $\bvzeta$ collects all primal and dual optimization variables, the implicit function theorem \cite{amannAnalysisII2008} can be applied to obtain the gradient if it exists.
For these cases, automatic differentiation tools have been proposed \cite{agrawalDifferentiableConvexOptimization2019, blondelEfficientModularImplicit2022}.
\par
Applying such implicit differentiation to Problem~\eqref{eq:problem_ul_inexactcsi_digital} is not straightforward due to its nonconvexity.
It is possible, however, to resort to the established convex reformulation of its dual downlink problem that leverages the semidefinite relaxation of $p_i \vb_i \vb_i^\He = \bvb_i \bvb_i ^\He$ to $ \breve{\bB}_i$, where $\bvb_i = \sqrt{p_i}\vb_i$ \cite{gershmanConvexOptimizationBasedBeamforming2010}:
\begin{equation}
	\begin{alignedat}{2}
		\min_{\oset{\breve{\bB}_i}_{i=1}^{I}} &&& \sum_{i=1}^{I}\trace*{\hbPsi_{0, i}\breve{\bB}_i}\\
		\quad \upm{s.t.} ~~ &&& (\forall i)\; \breve{\bB}_i \succeq \nullvec,\\
		&&& (\forall i)\; \gamma_i^{-1}\trace*{\hbPsi_{i,i}\breve{\bB}_i} \!-\! \sum_{j\neq i} \trace*{\hbPsi_{i,j}\breve{\bB}_j} \!\geq 1.%
	\end{alignedat} \label{eq:problem_dl_convex}
\end{equation}
\par
The corresponding first-order stationary condition can be absorbed into the dual variables $\bLambda_i$ for all $i$, that are associated with the \gls{psd} constraints in \eqref{eq:problem_dl_convex}, as 
\begin{equation}
	\bLambda_i(\vq, \uhvpsi) = \hbPsi_{0, i} - \gamma_i^{-1} q_i\hbPsi_{i,i} + \sum_{j\neq i} q_j \hbPsi_{j,i}, \label{eq:rankdef_interf_mat}
\end{equation}
where the virtual uplink power $q_i$ is the dual variable corresponding to the $i$th inequality constraint.
Leaving dual feasibility implicit, \ie, $\bLambda_i(\vq, \uhvpsi) \succeq \nullmat$ and $q_i\geq 0$, the \gls{kkt} conditions thus furnish the following \gls{snle}:
\begin{alignat}{2}
	(\forall i)\,&& \nullmat &= \bLambda_i(\vq, \uhvpsi) \breve{\bB}_i, \label{eq:kkt_convex_psd}\\
	(\forall i)\,&& 0 &= q_i\bigg(\! 1 - \gamma_i^{-1}\! \trace*{\!\hbPsi_{i,i} \breve{\bB}_i\!} \!+\! \sum_{j\neq i}\trace*{\!\hbPsi_{i,j} \breve{\bB}_j\!} \! \bigg) .\!\label{eq:kkt_convex_sinr}
\end{alignat}
\par
Albeit the implicit function theorem is now applicable if \eqref{eq:kkt_convex_psd}-\eqref{eq:kkt_convex_sinr} is formulated in terms of real values, the semidefinite relaxation substantially inflates the number of variables and, as a consequence, the cost of computing the gradient. 
In the following, to mitigate this problem, we therefore reduce \eqref{eq:kkt_convex_psd}-\eqref{eq:kkt_convex_sinr} into a suitable and more efficient representation.
To begin with, note that if a solution $(\breve{\bB}_i^\opt, q_i^\opt)_i$ to Problem \eqref{eq:problem_dl_convex} exists, where $\oset{q_i^\opt}_i$ are the corresponding dual variables, a solution with rank-1 $\breve{\bB}_i^\opt$ for all $i$ can always be found \cite[Th.~3.2]{huangRankConstrainedSeparableSemidefinite2010}. 
As such, the following equivalence is straightforward to show using the \gls{kkt} conditions in \eqref{eq:kkt_convex_psd}-\eqref{eq:kkt_convex_sinr}:
\begin{lemma} \label{lem:kkt_equivalence}
	 The point $\oset{\breve{\bB}_i^\opt, q_i^\opt}_i$ is a solution to Problem~\eqref{eq:problem_dl_convex} for some $\uhvpsi$ if and only if $(\breve{\vb}_i^\opt,q_i^\opt)_i$ with $\breve{\bB}_i^\opt= \bvb_i^\opt(\bvb_i^\opt)^\He$ is a dual feasible solution of the \gls{snle}
	 \begin{alignat}{2}
	 	(\forall i)\,&& \nullvec &= \vg_{\upm{b}, i}(\bvzeta, \uhvpsi) = \bLambda_i(\vq, \uhvpsi) \bvb_i,\label{eq:kkt_nonconvex_psd}\\
	 	(\forall i)\,&& 0 &= g_{\upm{q}, i}(\bvzeta, \uhvpsi) = q_i\bigg( \! 1 \!-\! \gamma_i^{-1} \bvb_i^\He\hbPsi_{i,i}\bvb_i + \! \sum_{j\neq i}\bvb_j^\He\hbPsi_{i,j} \bvb_j \!\bigg). \label{eq:kkt_nonconvex_sinr}\\[-8mm]\nonumber
	 \end{alignat}%
\end{lemma}
\par
Although \eqref{eq:kkt_nonconvex_psd}-\eqref{eq:kkt_nonconvex_sinr} retrieve the original optimization variables, they also reintroduce the phase ambiguity of the optimal beamformers $\vb_i^\opt$ in Problem~\eqref{eq:problem_ul_inexactcsi_digital}.
Specifically, if $\oset{\bvb_i^\opt, q_i^\opt}_{i=1}^{I}$ is a solution to the digital \gls{bf} problem, then $\oset{\bvb_i^\opt\eul^{\ju \phi_i}, q_i^\opt}_{i=1}^{I}$ is for arbitrary $\phi_i \in [0, 2\pi)$, resulting in a nonconvex solution set.
The existence of an infinite number of distinct solutions in an arbitrary open neighborhood of a solution $\vzeta = \vzeta^\opt$ to any differentiable \gls{snle} $\vg(\vzeta, \uhvpsi) = \nullvec$ implies that $\det\pdel*{\jcby[\vzeta]\vg(\vzeta^\opt, \uhvpsi)} = 0$ \cite[Corollary 7.8]{amannAnalysisII2008}.
Since this property prevents the application of the implicit function theorem, we further modify the \gls{snle} by i) restricting the set of solutions, ii) formulating it in terms of real-values and iii) removing redundant equations.
\par
First, we constrain the beamformer phase by setting $\matel{\Im(\bvb_i)}{\numrf} = 0$ and subsequently define
\begin{equation}
	\ubvzeta = \bmat*{\ubvb^\T ~~ \vq^\T}^\T = \bmat*{\ubvb_1^\T ~~ \cdots~~\ubvb_I^\T~~q_1~~\cdots~~q_I}^\T, \label{eq:def_realval_primdual_var}
\end{equation}
where $\ubvb_i = \bmat*{\Re(\bvb_i)^\T~~\matel{\Im(\bvb_i)}{\neg \numrf}^\T}^\T$ is the restricted real-valued beamformer.
Secondly and thirdly, we define the real-valued \gls{snle}
\begin{align}
	\nullvec &= \uvg(\ubvzeta, \uhvpsi) = \left[\uvg_{\upm{b}}(\ubvzeta, \uhvpsi)^\T ~~\uvg_{\upm{q}}(\ubvzeta, \uhvpsi)^\T \right]^\T, \label{eq:realval_kkt}
\end{align}
with the upper \gls{snle} corresponding to \eqref{eq:kkt_nonconvex_psd} with the imaginary part of the $\numrf$th element removed:
\begin{align}
	\uvg_{\upm{b}}(\ubvzeta, \uhvpsi) &= \left[\Re(\vg_{\upm{b}, 1}(\ubvzeta, \uhvpsi))^\T~~\matel{\Im(\vg_{\upm{b}, 1}(\ubvzeta, \uhvpsi))}{\neg \numrf}^\T \right.\nonumber\\
	&\quad~~\left.\cdots~ \Re(\vg_{\upm{b}, I}(\ubvzeta, \uhvpsi))^\T~~\matel{\Im(\vg_{\upm{b}, I}(\ubvzeta, \uhvpsi))}{\neg \numrf}^\T \right]^\T\!. %
\end{align}
The lower \gls{snle} is identical to \eqref{eq:kkt_nonconvex_sinr} and simply
\begin{equation}
	\uvg_{\upm{q}}(\ubvzeta, \uhvpsi) = \left[g_{\upm{q}, 1}(\ubvzeta, \uhvpsi) ~~ \cdots ~~ g_{\upm{q}, I}(\ubvzeta, \uhvpsi) \right]^\T,
\end{equation}
thereby treating the removed arguments $\matel{\Im(\bvb_i)}{\numrf}$ as $0$.
\par
Let us now identify the function returning a normalized, real-valued and optimal primal-dual solution of Problem~\eqref{eq:problem_ul_inexactcsi_digital} with added constraints $\matel{\Im(\vb_i)}{M} = 0$ for all $i$ as $\realulfun (\uhvpsi) = \uvzeta^\opt$.
In addition, we define the normalized primal-dual vector
\begin{equation}
	\uvf_\upm{BN} (\ubvzeta) = \bmat*{\breve{\uvf}_\upm{BN}^\T(\ubvb_1)~~\cdots~~ \breve{\uvf}_\upm{BN}^\T(\ubvb_I) ~~ \vq^\T }^\T \label{eq:bfnormfun}
\end{equation}
where $\underline{\vb} = \breve{\uvf}_\upm{BN}(\ubvb) = (\norm{\ubvb}_2)^{-1} \ubvb$.
Before stating the main result, we make the following technical assumption.
\begin{assumption}\label{asm:dlbf_unique_sol}
	The convex downlink \gls{bf} problem in \eqref{eq:problem_dl_convex} has a unique solution $\oset{\breve{\bB}_i^\opt, q_i^\opt}_i$ with $q_i^\opt>0$ for all $i$.
\end{assumption}
It can be easily observed from the dual problem in \eqref{eq:problem_ul_inexactcsi_digital} that if $\gamma_i>0$ and the problem is feasible, then $q_i > 0$.\footnote{In practice, feasible instances of \eqref{eq:problem_dl_convex} typically have unique solutions since eigenspaces of nonzero eigenvalues are almost never identical between channels, and principal eigenvalues almost never have multiplicity greater than one. Even in those cases, slight channel perturbations could be added.}

\begin{theorem} \label{th:dlbf_derivative}
	Let $\uvzeta^\opt = \realulfun (\uhvpsi)$ and $\ubvzeta^\opt$ such that $\uvzeta^\opt =\vf_\upm{BN}(\ubvzeta^\opt)$. 
	Under Assumption~\ref{asm:dlbf_unique_sol}, the point $\ubvzeta^\opt$ is a dual feasible critical point of $\uvg(\ubvzeta, \uhvpsi)$, the Jacobi matrix $\jcby[\uvzeta]{\uvg(\ubvzeta^\opt, \uhvpsi)}$ is full-rank and
	\begin{align}
		&\jcby[\uvpsi_{i,j}] \realulfun (\uhvpsi) = \nonumber\\[-1mm]
		&~~~~ -\jcby[\ubvzeta] f_\upm{BN}(\ubvzeta^\opt) \left(\jcby[\uvzeta]{\uvg(\ubvzeta^\opt, \uhvpsi)}\right)^{\!-1} \jcby[\uvpsi_{i,j}]{\uvg(\ubvzeta^\opt, \uhvpsi)}. \label{eq:implicit_derivative}
	\end{align}
\end{theorem}
\textit{Proof:} See Appendix~\ref{sec:apx_derivative_proof}.
\par 
A detailed description of the involved Jacobi matrices $\jcby[\uvpsi_{i,j}]{\uvg(\ubvzeta, \uhvpsi)}$ and $\jcby[\ubvzeta] f_\upm{BN}(\ubvzeta)$ is relegated to Appendix~\ref{sec:apx_jacobi}.
Theorem~\ref{th:dlbf_derivative} establishes $\uvg(\ubvzeta, \uhvpsi)$ as an admissible compact representation of Problem~\eqref{eq:problem_ul_inexactcsi_digital}. %
Compared with a direct implicit function approach based on the \gls{kkt} conditions \eqref{eq:kkt_convex_psd}-\eqref{eq:kkt_convex_sinr} of the convexified \gls{bf} optimization problem \eqref{eq:problem_dl_convex}, the proposed derivative substantially reduces the dimension of the inverted Jacobi matrix from $I(\numrf^2 + 1) \times I(\numrf^2 + 1)$ to $2I\numrf\times 2I\numrf$.
Furthermore, Theorem~\ref{th:dlbf_derivative} ensures the existence of the gradient provided a minor technical assumption, which is not necessarily the case when differentiating general convex optimization problems \cite{agrawalDifferentiableConvexOptimization2019}.
In practice, if the assumption is not met, \eg, when the problem is not feasible, we set the partial gradient to $\nullvec$.

\subsection{Training} \label{sec:training}
\subsubsection{Empirical risk optimization}
We optimize the \gls{dn} model parameters $\vtheta$ \wrt the objective of robust energy-efficient hybrid \gls{bf}, as defined by Problem~\eqref{eq:problem_dl_inexactcsi}, using empirical risk minimization.
While the power constraints are enforced by the projection \eqref{eq:pow_projection} embedded in the \gls{dn} architecture, an analogous projection is not applicable to the probabilistic \gls{qos} constraints.
Instead, the constraint can be integrated into the loss function as a penalty by following a primal-dual optimization approach \cite{nandwaniPrimalDualFormulation2019}.
\par
Introducing dual variables $\vlambda = \bmat{\lambda_1~ \cdots~ \lambda_{D_\upm{c}}}^\T$ for $D_\upm{c}$ constraints, the empirical risk can be composed as the Lagrangian
\begin{align}
	J(\dataset, \vtheta, \vlambda) =  \emexpect[\ssample \sim \dataset]{\weightvec^\T\vp_\cM({\cS}; \vtheta)} \!+\! \sum_{d=1}^{D_\upm{c}} \!\lambda_d \widetilde{g}(\dataset_d, \vtheta),\! \label{eq:lossfun}
\end{align}
where $\dataset$ is a dataset of system instances $\cS$.
Here, $\widetilde{g}(\dataset_d, \vtheta)$ denotes a differentiable empirical approximation of the \gls{qos} constraint \eqref{eq:prob_qos_constraint}, which we discuss later in this section. 
We allow each constraint term to be evaluated on a non-overlapping subset $\dataset_d \subseteq\dataset$, where $\dataset = \cup_{d=1}^{D_\upm{c}}\dataset_d$.
This provides additional flexibility in enforcing fairness across heterogeneous data subsets in the implementation of the constraint by the model $\cM$.
For instance, if the data set contains ``difficult'' and ``easy'' instances $\cS$ that require high and low transmit power, respectively, the trained model $\cM(\,\cdot\,; \vtheta^\opt)$ might satisfy the nominal outage $\pout$ on average over the entire data set, yet systematically violate the \gls{qos} requirement on the subset of difficult instances.
Partitioning the data into appropriately chosen subsets mitigates this imbalance by allowing the constraint to be applied more uniformly across different types of instances.
\par 
The model parameters $\vtheta$ are optimized by approaching the saddle point
\begin{align}
	\min_{\vtheta} \max_{\vlambda \geq \nullvec} ~J(\dataset, \vtheta, \vlambda)
\end{align}
via stochastic gradient descent steps based on minibatches $\breve{\dataset} \subset \dataset$.
Thereby, we alternate between primal descent steps in the direction $-\nabla_{\vtheta} J(\breve{\dataset}, \vtheta, \vlambda)$ and dual ascent steps in the direction $\nabla_{\vlambda} J(\breve{\dataset}, \vtheta, \vlambda)$ \cite{nandwaniPrimalDualFormulation2019}.

\subsubsection{Differentiable Outage Constraints}

\par
We follow an approach inspired by soft-to-hard annealing \cite{agustssonSofttoHardVectorQuantization2017} that we initially proposed in our preliminary work \cite{schynolCodebookBasedDownlinkBeamforming2024} to obtain a differentiable probabilistic constraint.
In particular, we approximate $g_i (\cM, \prdf{\ssample})$ in \eqref{eq:prob_qos_constraint} as
	\begin{align}
		&g_i \pdel*{\cM(\ssample), \prdf{\ssample}}\nonumber\\
		&\quad=1 - \pout\nonumber\\
		&\qquad -\expect[\cS \sim \prdf{\cS}]{\stepfun\left(\gamma_i^{-1}\sinrdl_{i} \left(\cM({\cS}, \vtheta), \cS\right) - 1 \right)} \nonumber \allowdisplaybreaks\\
		&\quad \approx  1 - \pout\nonumber\\
		&\qquad -\frac{1}{I}\sum_i \emexpect[\cS \sim \dataset]{ \apxstepfun_{\beta_\upm{c}} \left(\gamma_i^{-1}\sinrdl_{i} \left(\cM({\cS}, \vtheta), \cS\right) - 1 \right)}\nonumber\\
		&\quad = \widetilde{g}_{\beta_\upm{c}}(\dataset, \vtheta), \label{eq:qospenalty_stepapx}
	\end{align}
where the non-outage probability is first expressed in terms of an expectation as well as the unit step function $u(\cdot)$, then approximated by the empirical mean and the differentiable logistic function $\apxstepfun_{\beta_\upm{c}}(x) = (1 + \exp(-\beta_\upm{c}x))^{-1}$.
The parameter $\beta_\upm{c}$ controls the approximation of the unit step.
To ensure meaningful gradients for arbitrary \gls{qos} violations, $\beta_\upm{c}$ is adapted by exponential averaging with step size $\eta_\upm{c}$ in each training step $\tau$, \ie,
\begin{align}
	\beta_\upm{c}\idx{\tau} \!&=\! (1- \eta_\upm{c}) \beta_\upm{c}\idx{\tau - 1} \label{eq:qospenalty_stepapx_step}\\
	&+\!\frac{\eta_\upm{c}}{ \bdel*{ -\estqtl_{\ssample\sim\dataset} \Big( \gamma_i^{-1} \sinrdl_i\pdel*{ \cM\oset{\cS, \vtheta\idx{\tau}}} \!-\! 1; \pout\Big) }_{\!\overline{\beta}_\upm{c}^{-1}\!}^{\!\infty}}, \nonumber
\end{align}
where $\estqtl_{\ssample\sim\dataset}\pdel{f(\ssample); p}$ denotes the empirical $p$-quantile of $f(\ssample)$ based on data set $\dataset$.
As the outage probability decreases to the nominal value $\pout$, the empirical quantile in \eqref{eq:qospenalty_stepapx_step} approaches $0$, thus $\apxstepfun_{\beta_\upm{c}} (\cdot)$ approaches the step function until the upper bound $\beta_\upm{c}=\overline{\beta}_\upm{c}> 0$.
If the constraint is oversatisfied, $\beta_\upm{c}$ converges to $\overline{\beta}_\upm{c}$ as well.
Compared to \cite{jangDeepLearningApproach2022}, for instance, where the smoothing parameter $\beta_\upm{c}$ is fixed, the proposed adaptive approach more accurately approximates the original outage constraint, \eg, at the end of training, while simultaneously obtaining informative gradients if the constraint is significantly violated, \eg, at the start of training. 
As a result, more ``difficult'' \gls{sinr} requirements are realizable.
\par
The resulting gradient also differs from that of the conventional quantile–based approach for empirical risk minimization \cite{adamMachineLearningApproach2019}.
The quantile-based approach, as adopted in \cite{psomasDesignAnalysisSWIPT2022, yingDeepLearningBasedJoint2024, youDataAugmentationBased2021} for wireless \gls{qos} constraints, approximates the \gls{qos} constraint directly by the empirical quantile:
\begin{equation}
	\widetilde{g}(\dataset, \vtheta) = \estqtl_{\ssample\sim\dataset}\pdel*{\gamma_i \!-\! \sinrdl_i\pdel*{ \cM\oset{\cS, \vtheta}}; 1 \!-\!\pout}. \label{eq:qospenalty_quantile}
\end{equation}
This quantile is positive whenever the empirical outage probability exceeds the nominal outage and negative if the \gls{qos} requirement is met with a margin.
Here, the gradient depends only on the interpolated samples closest to the value of the empirical quantile \cite{adamMachineLearningApproach2019}.
In comparison, although our adaptive annealing approach introduces two hyperparameters $\eta_\upm{c}$ and $\overline{\beta}_\upm{c}$, it has the crucial advantage to construct the gradient by averaging more samples of $\sinrul_i$, thereby smoothing the loss surface.
We empirically compare the performance of both approaches in Sec.~\ref{sec:results}-\ref{sec:constraint_annealvsqtl}.
\subsubsection{Convergence Metric}
Model convergence is evaluated based on validation data $\dataset_{\upm{val}}$.
Over the course of training, the sequence of empirical loss values in \eqref{eq:lossfun} is not expected to be minimized near a saddle point because the dual updates continually reshape the objective landscape.
Nevertheless, a quantitative criterion for convergence is crucial to compare training runs and methods.
Therefore, we leverage the metric
\begin{align}
	J_{\upm{cm}}\idx{\tau} &= \frac{1}{I} \emexpect[\ssample \sim \valset]{ \weightvec^\T\vp_\cM({\cS}; \vtheta\idx{\tau})} \nonumber\\
	 & \quad + \sum_{d=1}^{D_\upm{c}} \overline{\lambda}_d \left[\what{g}\oset*{\cM(\cdot; \vtheta\idx{\tau}), \dataset_{\upm{val},d}}\right]_0^{\infty}, \label{eq:stopping_metric}
\end{align}
where $\what{g}$ is the empirical estimate of the constraint in \eqref{eq:prob_qos_constraint} on the subset $\dataset_{\upm{val},d}$ and averaged across users $i$, while $\overline{\lambda}_d$ is an upper bound of $\lambda_d$.
Negative values of $\what{g}$ are projected to $0$ to avoid rewarding constraint oversatisfaction.
Since all constraints other than the \gls{qos} constraints are implicitly guaranteed by the model architecture of $\cM$, $J_{\upm{cm}}\idx{\tau}$ is an upper bound of the Lagrangian function belonging to Problem \eqref{eq:problem_dl_inexactcsi} given $\lambda_d \leq \overline{\lambda}_d$. 
The bound is tight if the \gls{qos} constraints are exactly met.
$J_{\upm{cm}}\idx{\tau}$ is similar to the loss function of a penalty-based approach \cite{jangDeepLearningApproach2022}.
\par
Primal and dual step sizes are decayed once by $\eta_{\upm{a}}$ if the sequence $\oset{J_\upm{cm}\idx{\tau_\upm{val}\kappa}}_{\kappa}$ over training steps $\tau=\tau_\upm{val}\kappa$, where $\tau_\upm{val}$ is the validation interval, does not improve to a new minimum within a window of length
$\tau_\upm{pat,1}$. 
A run is considered converged if no improvement occurs within a window of length $\tau_\upm{pat,2}$.

\subsection{Dicussion and Limitations}
Compared to \cite[Alg.~2]{hegdeHybridBeamformingLargescale2017} for perfect \gls{csi}, the additional computational cost by the proposed \gls{dn} for imperfect \gls{csi}, that stems from the the instance-adaptive channel mapping, is marginal.
During validation, the cost remains dominated by the optimization of the low-dimensional digital beamformers $\bB\idx{\ell, a}$ across the \gls{cb}.
With the reasonable assumption $I < \numrf$, a complexity $\ocplx{L_\upm{ul} I{\numrf}^3}$ for solving \eqref{eq:problem_ul_inexactcsi_digital}, where $L_\upm{ul}$ accounts for iterations of the nested optimization, and ignoring the cost incurred by projections using the analog beams $\bva_i$, this computational complexity scales according to $\ocplx{\numaly L_\upm{ul} \numcb I {\numrf}^3}$.
As such, while the codebook-based hybrid architecture allows for significant design flexibility, it is traded with computational cost that scales linearly in $\numaly \geq \numrf$ and $\numcb$.
Similarly, since the cost of the gradient computation during training is dominated by the inversion required by Thm.~\ref{th:dlbf_derivative}, its computational complexity is $\ocplx{\numaly \numcb I^3 {\numrf}^3}$.
Due to the implicit differentiation, however, only the result of each instance of the digital \gls{bf} subproblem needs to be stored in memory instead of all intermediate results of the solver utilized for the subproblem.
\par
Furthermore, we do not make assumptions on the distribution of the \gls{csi} error.
Instead, the network learns both the effect of the \gls{csi} error as well as its relation to the feature $\vxi$ during training.
In this way, the design preserves the modular architecture of wireless systems by being agnostic to the \gls{csi} acquisition and feedback mechanism, unlike end-to-end joint estimation and \gls{bf} schemes such as the one in \cite{huTwoTimescaleEndtoEndLearning2022}.
Therefore, the proposed method is suitable for both \gls{mu} time domain duplex and frequency domain duplex systems, with \gls{csi} feedback compression being a major source of error in the latter.
In addition, the proposed architecture inherits support for the usage of both statistical and instantaneous \gls{csi}.
\par
There are two main limitations regarding the training procedure.
First, during training, the loss function requires access to accurate \gls{csi} to determine the outage, which can be challenging when working with real-world data.
However, the data-efficiency and strong generalization capability of the model-aided \gls{dn} approach help mitigating the problem.
Secondly, the feasibility of the training data must be carefully considered.
If an outage constraint renders the problem infeasible for the given combination of data and model architecture, \eg, if the set nominal outage is too low given the number of users, nominal \gls{sinr}, or \gls{csi} uncertainty, then the corresponding dual variable grows unbounded and the training diverges. 
As such, the composition of the training data set must be designed with particular care.
During testing, however, if the inner baseband problem \eqref{eq:problem_ul_inexactcsi_digital} is infeasible for a system realization, \eg, if \eqref{eq:dlpow_calc_apx} yields negative power allocations, a failure to satisfy the \gls{qos} constraints is clearly indicated.
In these cases, the problem instance could be relaxed by, \eg, rescheduling users to another resource block \cite{tajallifarRobustFeasibleQoSAware2024}.

\section{Empirical Results}\label{sec:results}

\subsection{Experimental Setup} \label{sec:setup}
\begin{table}
	\setlength{\tabcolsep}{4pt}
	\caption{Parameters used for training, model configuration and data.\label{tab:simdata_param}}
	\begin{tabularx}{\linewidth}{Xcr}
		\toprule
		\textbf{Training Param.} &  \textbf{Symbol} & \textbf{Value}\\
		\midrule
		Adam smoothing &  & $(0.9, 0.99)$\\ %
		Adam weight decay & & $0.1$\\
		primal step size & $\eta_\upm{p}$ & $0.0005$\\
		dual step size & $\eta_\upm{d}$ & $0.1$ ($0.002$ for \eqref{eq:qospenalty_quantile})\\
		step size decay & $\eta_\upm{a}$ & $0.2$ \\
		minibatch size & $\abs{\breve{\dataset}}$ & $200$ ($10$ in Sec~\ref{sec:results}-\ref{sec:exp_statinst})\\
		anneal. constraint &  $(\overline{\beta}_\upm{c}, \eta_{\upm{c}})$ & $(50, 0.01)$ \\
		validation interval & $\tau_\upm{val}$ & $ 100$ \\ 
		dual upper bound & $\overline{\lambda}_d$ & $100$ \\
		stop. crit. patience & $(\tau_\upm{pat,1}, \tau_\upm{pat,2})$ & $(5000, 10000)$\\
		\# train. instances & $\abs{\trainset}$ & 20000\\
		\# valid. instances & $\abs{\valset}$ & 5000\\
		\midrule
		\textbf{Model Param.} &  \textbf{Symbol} & \textbf{Value}\\
		\midrule
		\# greedy selections & $\numaly$ & $\numrf$\\
		beam annealing & $\beta_\upm{M}$ & 5\\
		\gls{gcnn} param.  & \!\!\!$(\numgnnly, \numgnnfilt, F\idx{\ell}\!) $ & $(3, 1, 32)$\\
		coefficient bound & $\overline{\beta}_\upm{0}$& $8$\\
		\midrule
		\textbf{Data Param.} &  \textbf{Symbol} & \textbf{Value}\\
		\midrule
		Tx Antenna Dim. & $(M_\upm{x},M_\upm{y})$ & $(4, 4)$\\
		angular direction & $(\varphi_{\upm{x},i}, \varphi_{\upm{y},i})$ & $\in [-60^{\circ}\!,60^{\circ}]\! \times\! [-60^{\circ}\!, 30^{\circ}]$\\
		angular spread & $\sigma_\upm{as}$ & \qty{10}{\degree} (\qty{5}{\degree} in Sec.~\ref{sec:exp_statinst})\\
		max. Tx power & $\pmax$ & \qty{20}{\decibel}\\
		minimum \gls{sinr} & $\gamma_i$ & $\in[5\unit{\decibel}, 15\unit{\decibel}]$\\
		\# RF chains & $\numrf$ & 5\\
		nominal outage & $\pout$ & 0.1\\
		\bottomrule
	\end{tabularx}
\end{table}

The proposed architecture, training and experiments are implemented using PyTorch.
To promote reproducibility, our code is published at \url{https://github.com/lsky96/outage-constrained-hybrid-bf-md-dnn}. %
\subsubsection{Training} %
\Gls{dn} parameters are learned with minibatch stochastic gradient descent and the primal gradients are normalized using Adam with weight decay \cite{loshchilovDecoupledWeightDecay2019}.
We employ $5$-fold cross-validation and results are given as mean $\pm$ standard deviation.
To prevent convergence problems due to sporadic gradient explosion, we employ adaptive gradient clipping \cite{seetharamanAutoclipAdaptiveGradient2020}.
Specifically, we clip the $\infty$-norm of the primal gradient $\norm{\nabla_{\vtheta} J}_\infty$ to a maximum of the current $0.90$-quantile of the gradient norm history.
Unless stated otherwise, the training parameters are summarized in Tab.~\ref{tab:simdata_param} and the loss function in \eqref{eq:lossfun} with the annealing-based constraint in \eqref{eq:qospenalty_stepapx} are utilized.
Pseudocode of the training procedure using the proposed annealing-based constraint is provided in Appendix~\ref{sec:training_proc}.
We remark that we found in empirical experiments that different choices of $\overline{\lambda}_d$ do not significantly impact the final model performance as long as $\overline{\lambda}_d$ is chosen as an upper bound of the dual variables.

\subsubsection{Configuration of the proposed method} 
\par
For the proposed architecture in Sec.~\ref{sec:prop_robust_greedybf}-\ref{sec:hybrid_architecture}, which we henceforth denote as \textit{U-G-HBF-GCN}, we utilize $\numgnnly$-layer \glspl{gcnn} of degree $\numgnnfilt=1$, where the nonlinear activation functions $\phi\idx{\ell}$ are $\relu$, except for the final layer where $\phi\idx{\numgnnly}(x) = \eul^{\overline{\beta}_\upm{o} \tanh(x / \overline{\beta}_\upm{o})}$ is an exponential with soft bounds $[\eul^{-\overline{\beta}_\upm{o}}, \eul^{\overline{\beta}_\upm{o}}]$ \cite{schynolAdaptiveAnomalyDetection2025}.
The input features are normalized by a batchnorm layer \cite{ioffeBatchNormalizationAccelerating2015}.
A constant annealing parameter $\beta_\upm{M}$ is chosen in \eqref{eq:beam_selection_apx} for simplicity.

\subsubsection{Benchmark methods} We compare to four approaches:
\begin{itemize}
	\item \textit{G-HBF-perf} is the greedy hybrid \gls{bf} algorithm \cite[Alg.~2]{hegdeHybridBeamformingLargescale2017} with $\numaly=2\numrf$ selections and perfect \gls{csi}.
	\item \textit{G-HBF-marg} is the greedy hybrid \gls{bf} algorithm \cite[Alg.~2]{hegdeHybridBeamformingLargescale2017} accessing imperfect CSI with \gls{sinr} target $\gamma_i' = \gamma_i + \breve{\gamma}$, where $\breve{\gamma}$ is bisected until the nominal outage $\pout$ is approximately attained on the validation data.
	\item \textit{U-G-HBF-FCN} is similar to U-G-HBF-GCN, but uses only a \gls{gcnn} of degree $\numgnnfilt=0$ with increased $F\idx{\ell}=64$, which is equivalent to parallel fully connected \glspl{dn} for each user.
	\item \textit{U-FDBF-GCN} is a fully-digital version of U-G-HBF-GCN, where $\bA = \idmat$ effectively.%
\end{itemize}
Both G-HBF-perf and U-FDBF-GCN serve as bounds that cannot be achieved without full \gls{csi} or per-antenna \gls{rf} chains at the transmitter, respectively.

\subsubsection{Data Generation}
\par
For each system instance $\cS$, we generate channel covariance matrices $\bR_i=\bR_{\upm{x},i} \kron \bR_{\upm{y},i} \in \Cset^{M\times M}$ with $M=M_\upm{x}M_\upm{y}$ of a 2D uniform planar array with half-wavelength antenna distance.
The covariance matrices $\bR_{\upm{x},i}\in\Cset^{M_\upm{x}\times M_\upm{x}}$ and $\bR_{\upm{y},i}\in\Cset^{M_\upm{y}\times M_\upm{y}}$ of both axes follow the spatial correlation model in \cite{forenzaSimplifiedSpatialCorrelation2007} with uniformly sampled angular directions $(\varphi_{\upm{x},i}, \varphi_{\upm{y},i})$ of user $i$.
Instantaneous channel vectors are then sampled as $\vh_i\sim\cnormdistr{\nullvec, \bR_i}$.
We remove instances that are infeasible under perfect \gls{csi}.
\par
We subsequently degrade the \gls{csi} in two ways:
\begin{itemize}
	\item \textit{Noisy MMSE estimator model}: We follow \cite{bjornsonMassiveMIMOImperfect2016} and use %
	\begin{equation}
		\hvh_i = \bR_i\left(\!\bR_i + \ppilot^{-1} \idmat\right)^{\!-1} \!\left(\vh_i + \sqrt{\ppilot}^{-1}\vn_i \right)\!,
	\end{equation}
	where $\vn_i \sim \cnormdistr{\nullvec, \idmat}$ and $\ppilot = \xi_i$ is the power of the pilot signal.
	For simplicity, the MMSE estimate accesses the true covariance.
	\item \textit{DFT feedback quantization model}: $\hvh_i$ is obtained by 2D \gls{dft} \gls{cb} feedback quantization with $\numfbv = \xi_i$ feedback vectors modeled after the 3GPP 5G Type II codebook \cite{dreifuerstMachineLearningCodebook2024, 3gppPhysicalLayerProcedures2022}, but restricted to a narrowband and unpolarized setting.
	The channel $\vh_i$ is first transformed by a $2\times 2$-times oversampled 2D \gls{dft}.
	For each set of \gls{dft}-coefficients corresponding to one of the $4$ orthogonal \gls{dft} vector subsets of the oversampled 2D-\gls{dft}, the top $N_\upm{fb}$ coefficients are quantized in $\qty{3}{\dB}$ magnitude steps and $8$-PSK phase steps. 
	$\hvh_i$ is obtained by applying the IDFT to the top $N_\upm{fb}$ quantized coefficients of the oversampled vector subset that minimizes $\norm{\tvh_i}_2$.
\end{itemize}
\par 
The system instances $\ssample$ are divided into $D_\upm{c}$ constraint groups as explained in Sec.~\ref{sec:prop_robust_greedybf}-\ref{sec:training}, depending on the degradation $\xi_i$, \ie, $\vxi_i$ is scalar.

\subsubsection{Transmit codebook} 
The transmit \gls{cb} $\cbook$ is a 2D-\gls{dft} \gls{cb} without oversampling, leading to a size $\numcb=M$ \cite{shiDeepLearningBasedRobust2021, dreifuerstMachineLearningCodebook2024}.

\subsection{Comparison between Methods} \label{sec:4x4_comparison}
\begin{figure*}
	\centering
	\pgfplotsset{
	table/search path={journal/data/algcomp}
}
\edef\pval{10}
\def\dpdata{eval_inst_1bs4x4_3ue1_snr0_as10_pmax20_qos5-15_rf5_mn_p\pval}
\def\dpGCNN{HybridU_msa_logconv_4i_3l_f2x32_bn_imperfCSI}
\def\dpGCNNqtl{HybridU_msa_logconv_4i_3l_f2x32_bn_imperfCSI_qtl}
\def\dpFCNN{HybridU_msa_fcnn_4i_3l_f64_bn_imperfCSI}
\def\dpSolveImperf{hybrid_TX_solve_imperfCSI_bisectmargin}

\tikzset{external/export next=false} %
\begin{tikzpicture}
\begin{groupplot}[
	group style={
		group name=algcomp, 
		group size=4 by 1,
		horizontal sep=0.7cm, 
		vertical sep=0.4cm,
		}, 
		height=5cm, 
		width=0.291\linewidth, 
		xlabel near ticks, 
		ylabel near ticks,
		cycle list={},
		every axis plot/.append style={only marks, mark=*, mark size=1.2mm, opacity=0.5, draw opacity=0},
		every axis title/.append style={yshift=-2.5mm},
		]
	
	\nextgroupplot[
	title={$\ppilot \in [10\unit{\decibel},24\unit{\decibel}]$},
	xmin=8, xmax=12,
	xtickmin=0,
	ymin=5, ymax=9.5,
	xlabel={Out. prob.  $\estpout$ [\%]},
	ylabel = {Avg. Power $\weightvec^\T \vp$},
	]
	\edef\pval{10-24}
	\addplot+[mark options={fill=darkRed}] table[col sep=tab, x expr=\thisrowno{2}*100, y index=1]{pout09/\dpdata_\dpGCNN.csv};
	\addplot+[mark options={fill=darkRed}, mark=triangle*, forget plot] table[col sep=tab, x expr=\thisrowno{2}*100, y index=1]{pout10/\dpdata_\dpGCNN.csv};
	\addplot+[mark options={fill=darkRed}, mark=square*, forget plot] table[col sep=tab, x expr=\thisrowno{2}*100, y index=1]{pout11/\dpdata_\dpGCNN.csv};
	
	\addplot+[mark options={fill=darkBlue}] table[col sep=tab, x expr=\thisrowno{2}*100, y index=1]{pout09/\dpdata_\dpFCNN.csv};
	\addplot+[mark options={fill=darkBlue}, mark=triangle*, forget plot] table[col sep=tab, x expr=\thisrowno{2}*100, y index=1]{pout10/\dpdata_\dpFCNN.csv};
	\addplot+[mark options={fill=darkBlue}, mark=square*, forget plot] table[col sep=tab, x expr=\thisrowno{2}*100, y index=1]{pout11/\dpdata_\dpFCNN.csv};

	\addplot+[mark options={fill=black}] table[col sep=tab, x expr=\thisrowno{2}*100, y index=1]{pout09/\dpdata_\dpSolveImperf.csv};
	\addplot+[mark options={fill=black}, mark=triangle*, forget plot] table[col sep=tab, x expr=\thisrowno{2}*100, y index=1]{pout10/\dpdata_\dpSolveImperf.csv};
	\addplot+[mark options={fill=black}, mark=square*, forget plot] table[col sep=tab, x expr=\thisrowno{2}*100, y index=1]{pout11/\dpdata_\dpSolveImperf.csv};
	
	\addplot+[mark options={color=darkRed}, mark=o, draw opacity=0.5] table[col sep=tab, x expr=\thisrowno{2}*100, y index=1]{pout09/\dpdata_\dpGCNNqtl.csv};
	\addplot+[mark options={color=darkRed}, mark=triangle, draw opacity=0.5, forget plot] table[col sep=tab, x expr=\thisrowno{2}*100, y index=1]{pout10/\dpdata_\dpGCNNqtl.csv};
	\addplot+[mark options={color=darkRed}, mark=square, draw opacity=0.5, forget plot] table[col sep=tab, x expr=\thisrowno{2}*100, y index=1]{pout11/\dpdata_\dpGCNNqtl.csv};

	\nextgroupplot[
	title={$\ppilot = 10\unit{\decibel}$},
	xmin=8, xmax=12,
	xtickmin=0,
	ymin=9, ymax=17.5,
	xlabel={Out. prob. $\estpout$ [\%]},
	]
	\edef\pval{10}
	\addplot+[mark options={fill=darkRed}] table[col sep=tab, x expr=\thisrowno{2}*100, y index=1]{pout09/\dpdata_\dpGCNN.csv};
	\addplot+[mark options={fill=darkRed, mark=triangle*}, forget plot] table[col sep=tab, x expr=\thisrowno{2}*100, y index=1]{pout10/\dpdata_\dpGCNN.csv};
	\addplot+[mark options={fill=darkRed}, mark=square*, forget plot] table[col sep=tab, x expr=\thisrowno{2}*100, y index=1]{pout11/\dpdata_\dpGCNN.csv};
	
	\addplot+[mark options={fill=darkBlue}] table[col sep=tab, x expr=\thisrowno{2}*100, y index=1]{pout09/\dpdata_\dpFCNN.csv};
	\addplot+[mark options={fill=darkBlue}, mark=triangle*, forget plot] table[col sep=tab, x expr=\thisrowno{2}*100, y index=1]{pout10/\dpdata_\dpFCNN.csv};
	\addplot+[mark options={fill=darkBlue}, mark=square*, forget plot] table[col sep=tab, x expr=\thisrowno{2}*100, y index=1]{pout11/\dpdata_\dpFCNN.csv};

	\addplot+[mark options={fill=black}] table[col sep=tab, x expr=\thisrowno{2}*100, y index=1]{pout09/\dpdata_\dpSolveImperf.csv};
	\addplot+[mark options={fill=black}, mark=triangle*, forget plot] table[col sep=tab, x expr=\thisrowno{2}*100, y index=1]{pout10/\dpdata_\dpSolveImperf.csv};
	\addplot+[mark options={fill=black}, mark=square*, forget plot] table[col sep=tab, x expr=\thisrowno{2}*100, y index=1]{pout11/\dpdata_\dpSolveImperf.csv};
	
	\addplot+[mark options={color=darkRed}, mark=o, draw opacity=0.5] table[col sep=tab, x expr=\thisrowno{2}*100, y index=1]{pout09/\dpdata_\dpGCNNqtl.csv};
	\addplot+[mark options={color=darkRed}, mark=triangle, draw opacity=0.5, forget plot] table[col sep=tab, x expr=\thisrowno{2}*100, y index=1]{pout10/\dpdata_\dpGCNNqtl.csv};
	\addplot+[mark options={color=darkRed}, mark=square, draw opacity=0.5, forget plot] table[col sep=tab, x expr=\thisrowno{2}*100, y index=1]{pout11/\dpdata_\dpGCNNqtl.csv};

	\nextgroupplot[
	title={$\ppilot = 17\unit{\decibel}$},
	xmin=8, xmax=12,
	xtickmin=0,
	ymin=5, ymax=9,
	xlabel={Out. prob. $\estpout$ [\%]},
	]
	\edef\pval{17}
	\addplot+[mark options={fill=darkRed}] table[col sep=tab, x expr=\thisrowno{2}*100, y index=1]{pout09/\dpdata_\dpGCNN.csv};
	\addplot+[mark options={fill=darkRed}, mark=triangle*, forget plot] table[col sep=tab, x expr=\thisrowno{2}*100, y index=1]{pout10/\dpdata_\dpGCNN.csv};
	\addplot+[mark options={fill=darkRed}, mark=square*, forget plot] table[col sep=tab, x expr=\thisrowno{2}*100, y index=1]{pout11/\dpdata_\dpGCNN.csv};
	
	\addplot+[mark options={fill=darkBlue}] table[col sep=tab, x expr=\thisrowno{2}*100, y index=1]{pout09/\dpdata_\dpFCNN.csv};
	\addplot+[mark options={fill=darkBlue}, mark=triangle*, forget plot] table[col sep=tab, x expr=\thisrowno{2}*100, y index=1]{pout10/\dpdata_\dpFCNN.csv};
	\addplot+[mark options={fill=darkBlue}, mark=square*, forget plot] table[col sep=tab, x expr=\thisrowno{2}*100, y index=1]{pout11/\dpdata_\dpFCNN.csv};

	\addplot+[mark options={fill=black}] table[col sep=tab, x expr=\thisrowno{2}*100, y index=1]{pout09/\dpdata_\dpSolveImperf.csv};
	\addplot+[mark options={fill=black}, mark=triangle*, forget plot] table[col sep=tab, x expr=\thisrowno{2}*100, y index=1]{pout10/\dpdata_\dpSolveImperf.csv};
	\addplot+[mark options={fill=black}, mark=square*, forget plot] table[col sep=tab, x expr=\thisrowno{2}*100, y index=1]{pout11/\dpdata_\dpSolveImperf.csv};
	
	\addplot+[mark options={color=darkRed}, mark=o, draw opacity=0.5] table[col sep=tab, x expr=\thisrowno{2}*100, y index=1]{pout09/\dpdata_\dpGCNNqtl.csv};
	\addplot+[mark options={color=darkRed}, mark=triangle, draw opacity=0.5, forget plot] table[col sep=tab, x expr=\thisrowno{2}*100, y index=1]{pout10/\dpdata_\dpGCNNqtl.csv};
	\addplot+[mark options={color=darkRed}, mark=square, draw opacity=0.5, forget plot] table[col sep=tab, x expr=\thisrowno{2}*100, y index=1]{pout11/\dpdata_\dpGCNNqtl.csv};

	\nextgroupplot[
	title={$\ppilot = 24\unit{\decibel}$},
	xmin=8, xmax=12,
	xtickmin=0,
	ymin=4, ymax=8,
	xlabel={Out. prob. $\estpout$ [\%]},
	legend style={legend columns=5, at={(1,1.13)}},
	]
	\edef\pval{24}
	\addplot+[mark options={fill=darkRed}] table[col sep=tab, x expr=\thisrowno{2}*100, y index=1]{pout09/\dpdata_\dpGCNN.csv};
	\addplot+[mark options={fill=darkRed}, mark=triangle*, forget plot] table[col sep=tab, x expr=\thisrowno{2}*100, y index=1]{pout10/\dpdata_\dpGCNN.csv};
	\addplot+[mark options={fill=darkRed}, forget plot] table[col sep=tab, x expr=\thisrowno{2}*100, y index=1]{pout11/\dpdata_\dpGCNN.csv};
	
	\addplot+[mark options={fill=darkBlue}] table[col sep=tab, x expr=\thisrowno{2}*100, y index=1]{pout09/\dpdata_\dpFCNN.csv};
	\addplot+[mark options={fill=darkBlue}, mark=triangle*, forget plot] table[col sep=tab, x expr=\thisrowno{2}*100, y index=1]{pout10/\dpdata_\dpFCNN.csv};
	\addplot+[mark options={fill=darkBlue}, mark=square*, forget plot] table[col sep=tab, x expr=\thisrowno{2}*100, y index=1]{pout11/\dpdata_\dpFCNN.csv};

	\addplot+[mark options={fill=black}] table[col sep=tab, x expr=\thisrowno{2}*100, y index=1]{pout09/\dpdata_\dpSolveImperf.csv};
	\addplot+[mark options={fill=black}, mark=triangle*, forget plot] table[col sep=tab, x expr=\thisrowno{2}*100, y index=1]{pout10/\dpdata_\dpSolveImperf.csv};
	\addplot+[mark options={fill=black}, mark=square*, forget plot] table[col sep=tab, x expr=\thisrowno{2}*100, y index=1]{pout11/\dpdata_\dpSolveImperf.csv};
	
	\addplot+[mark options={color=darkRed}, mark=o, draw opacity=0.5] table[col sep=tab, x expr=\thisrowno{2}*100, y index=1]{pout09/\dpdata_\dpGCNNqtl.csv};
	\addplot+[mark options={color=darkRed}, mark=triangle, draw opacity=0.5, forget plot] table[col sep=tab, x expr=\thisrowno{2}*100, y index=1]{pout10/\dpdata_\dpGCNNqtl.csv};
	\addplot+[mark options={color=darkRed}, mark=square, draw opacity=0.5, forget plot] table[col sep=tab, x expr=\thisrowno{2}*100, y index=1]{pout11/\dpdata_\dpGCNNqtl.csv};

	\legend{U-G-HBF-GCN w/ \eqref{eq:qospenalty_stepapx_step}, U-G-HBF-FCN w/ \eqref{eq:qospenalty_stepapx_step}, G-HBF-marg, U-G-HBF-GCN w/ \eqref{eq:qospenalty_quantile}}

\end{groupplot}
\end{tikzpicture}
	\vspace{-3mm}
	\caption{Allocated power over empirical outage probability for different targets $\pout$: $0.09$ (circles), $0.10$ (triangles), $0.11$ (squares). 5 folds each. \label{fig:4x4_comparison}}
\end{figure*}

\begin{table*}
	\setlength{\tabcolsep}{4.0pt}
	\caption{Results for $M=4\times 4$ antennas with 3 users for nominal outage $\pout = 0.10$ for \gls{csi} degradation based on different $\ppilot$.\label{tab:4x4_comparison}}
	\begin{tabularx}{\linewidth}{Xrrrrrrrr}
		\toprule %
		& \multicolumn{2}{c}{$\ppilot \in [10\unit{\dB}, 24\unit{\decibel}]$} & \multicolumn{2}{c}{$\ppilot = 10\unit{\decibel}$} & \multicolumn{2}{c}{$\ppilot = 17\unit{\decibel}$} & \multicolumn{2}{c}{$\ppilot = 24\unit{\decibel}$}  \\
		\cmidrule(lr){2-3} \cmidrule(lr){4-5} \cmidrule(lr){6-7} \cmidrule(lr){8-9}
		\textbf{Methods} & $\norm{\vp}_1$ & $\estpout$ $[\%]$ & $\norm{\vp}_1$ & $\estpout$ $[\%]$ & $\norm{\vp}_1$ & $\estpout$ $[\%]$ & $\norm{\vp}_1$ & $\estpout$ $[\%]$\\
		\midrule
		G-HBF-perf & $6.61 \pm 0.10$ & $0.00 \pm 0.00$ & $6.62 \pm 0.17$ & $0.00 \pm 0.00$ & $6.62 \pm 0.11$ & $0.00 \pm 0.00$ & $6.70 \pm 0.04$ & $0.00 \pm 0.00$ \\
		G-HBF-marg & $8.86 \pm 0.17$ & $10.03 \pm 0.01$ & $15.44 \pm 0.40$ & $10.00 \pm 0.00$ & $8.41 \pm 0.17$ & $10.01 \pm 0.04$ & $7.27 \pm 0.05$ & $10.03 \pm 0.02$ \\
		U-G-HBF-FCN & $6.62 \pm 0.14$ & $9.93 \pm 0.55$ & $12.07 \pm 0.69$ & $10.59 \pm 0.57$ & $6.17 \pm 0.18$ & $9.86 \pm 1.03$ & $5.26 \pm 0.17$ & $9.28 \pm 0.49$ \\
		U-G-HBF-GCN & $6.28 \pm 0.07$ & $9.67 \pm 0.70$ & $11.04 \pm 0.20$ & $10.10 \pm 0.51$ & $5.92 \pm 0.17$ & $9.77 \pm 0.87$ & $5.14 \pm 0.15$ & $9.42 \pm 0.54$ \\
		U-FDBF-GCN & $3.98 \pm 0.07$ & $9.55 \pm 0.75$ & $6.04 \pm 0.15$ & $9.90 \pm 0.78$ & $3.83 \pm 0.13$ & $9.64 \pm 0.89$ & $3.45 \pm 0.00$ & $9.49 \pm 0.02$\\
		\midrule
		U-G-HBF-GCN w/ \eqref{eq:qospenalty_quantile} & $6.86 \pm 0.09$ & $8.83 \pm 0.43$ & $11.83 \pm 0.32$ & $9.74 \pm 0.30$ & $6.47 \pm 0.10$ & $9.10 \pm 0.83$ & $5.61 \pm 0.17$ & $8.04 \pm 0.44$\\
		\bottomrule
	\end{tabularx}
\end{table*}
We construct datasets according to the parameters summarized in Tab.~\ref{tab:simdata_param} with \gls{csi} degradation through noisy MMSE estimation with pilot power $\ppilot = \xi_i$.
The data is subdivided into 4 constraint groups with $\xi_i = 10\unit{\decibel}$, $\xi_i = 17\unit{\decibel}$, $\xi_i = 24\unit{\decibel}$ and $\xi_i$ uniformly distributed over $[\qty{10}{\decibel}, \qty{24}{\decibel}]$, respectively. 
We remark that for this setup, on a machine with Arch Linux and AMD Epyc 9554P CPU limited to 3 threads, the proposed U-G-HBF-GCN has a runtime of $1.9\pm 0.4 \unit{\milli \second}$ per system realization, while one training step has a runtime of approximately \qty{1}{\second}.
\par
The proposed and benchmark methods are compared in Fig.~\ref{fig:4x4_comparison} and Tab.~\ref{tab:4x4_comparison}.
To observe trends in the power-outage trade-off, we train models for nominal outage probabilities $\pout$ of $0.09$, $0.10$ and $0.11$, respectively.
The proposed learning-based methods significantly outperform the benchmark G-HBF-marg.
Moreover, U-G-HBF-GCN outperforms U-G-HBF-FCN, particularly if $\ppilot=\qty{10}{\dB}$ or if $\ppilot$ is uniformly distributed over $[{\qty{10}{\dB}, \qty{24}{\dB}}]$.
This demonstrates the efficacy of the \gls{gnn}-based approach, which leverages the similarity between channels of different users.
The performance of the methods converge as the \gls{csi} quality improves, \ie, as the $\ppilot$ increases, which is expected. %
Since a nonzero outage is targeted, the learning-based methods may allocate less power than G-HBF-perf.
The fully-digital U-FDBF-GCN achieves a significantly better power-outage trade-off due to the number of \gls{rf} chains of U-G-HBF-GCN, as will be seen in Sec.~\ref{sec:results}-\ref{sec:4x4_generalization}. %
Notably, we find that the proposed \gls{dn} models reliably generalize between different levels of \gls{csi} degradation, matching $\pout$ in most cases within $1\%$, thereby demonstrating the effectiveness of multiple constraints and side-information $\xi_i$.

\subsection{Comparison to Quantile-based Constraint} \label{sec:constraint_annealvsqtl}
\begin{figure}
	\centering
	\pgfplotsset{
	table/search path={journal/data/training}
}
\edef\pval{10}
\def\dpnoqtl{noqtl}
\def\dpqtl{qtl}

\tikzsetnextfilename{training_noqtl_qtl}

\begin{tikzpicture}
\begin{groupplot}[
	group style={
		group name=algcomp, 
		group size=2 by 2,
		horizontal sep=0.8cm, 
		vertical sep=0.4cm,
		}, 
		height=4.0cm, 
		width=0.54\linewidth, 
		xlabel style={align=center},
		xlabel near ticks, 
		ylabel near ticks,
		cycle list={},
		every axis plot/.append style={},
		every axis title/.append style={yshift=-7pt},
		]

	\nextgroupplot[
		title={annealing-based (prop.)},
		xmin=0, xmax=22000,
		xtickmin=0,
		xticklabels={},
		ymin=0, ymax=13,
		ylabel={Dual var. $\lambda_d$},
		scaled x ticks=false,
	]
	\addplot+[blue1] table[col sep=comma, x expr=\thisrowno{1}, y expr={\thisrowno{2}}]{\dpnoqtl/training_dual.csv};
	\addplot+[blue2] table[col sep=comma, x expr=\thisrowno{1}, y expr={\thisrowno{3}}]{\dpnoqtl/training_dual.csv};
	\addplot+[blue3] table[col sep=comma, x expr=\thisrowno{1}, y expr={\thisrowno{4}}]{\dpnoqtl/training_dual.csv};
	\addplot+[blue4] table[col sep=comma, x expr=\thisrowno{1}, y expr={\thisrowno{5}}]{\dpnoqtl/training_dual.csv};
	
	\nextgroupplot[
		title={quantile-based},
		xmin=0, xmax=41000,
		ymin=0, ymax=2.2,
		xtickmin=0,
		xticklabels={},
		scaled x ticks=false,
		legend style={legend columns=4, at={(1,1.15)}},
	]
	\addplot+[blue1] table[col sep=comma, x expr=\thisrowno{1}, y expr={\thisrowno{2}}]{\dpqtl/training_dual.csv};
	\addplot+[blue2] table[col sep=comma, x expr=\thisrowno{1}, y expr={\thisrowno{3}}]{\dpqtl/training_dual.csv};
	\addplot+[blue3] table[col sep=comma, x expr=\thisrowno{1}, y expr={\thisrowno{4}}]{\dpqtl/training_dual.csv};
	\addplot+[blue4] table[col sep=comma, x expr=\thisrowno{1}, y expr={\thisrowno{5}}]{\dpqtl/training_dual.csv};
	
	\legend{{$\in[10\unit{\dB},24\unit{\dB}]$}, {$10\unit{\dB}$}, {$17\unit{\dB}$}, {$24\unit{\dB}$}}
	
	\nextgroupplot[
		xmin=0, xmax=22000,
		xtickmin=0,
		ymin=0, ymax=25,
		xlabel={Step $\tau$},
		ylabel={Out. prob. $\estpout$ [\%]},
	]
	\addplot+[blue1] table[col sep=comma, x expr=\thisrowno{1}, y expr={\thisrowno{2}*100}]{\dpnoqtl/validation_outage.csv};
	\addplot+[blue2] table[col sep=comma, x expr=\thisrowno{1}, y expr={\thisrowno{3}*100}]{\dpnoqtl/validation_outage.csv};
	\addplot+[blue3] table[col sep=comma, x expr=\thisrowno{1}, y expr={\thisrowno{4}*100}]{\dpnoqtl/validation_outage.csv};
	\addplot+[blue4] table[col sep=comma, x expr=\thisrowno{1}, y expr={\thisrowno{5}*100}]{\dpnoqtl/validation_outage.csv};
	
	\nextgroupplot[
		xmin=0, xmax=41000,
		xtickmin=0,
		ymin=0, ymax=25,
		xlabel={Step $\tau$},
	]
	\addplot+[blue1] table[col sep=comma, x expr=\thisrowno{1}, y expr={\thisrowno{2}*100}]{\dpqtl/validation_outage.csv};
	\addplot+[blue2] table[col sep=comma, x expr=\thisrowno{1}, y expr={\thisrowno{3}*100}]{\dpqtl/validation_outage.csv};
	\addplot+[blue3] table[col sep=comma, x expr=\thisrowno{1}, y expr={\thisrowno{4}*100}]{\dpqtl/validation_outage.csv};
	\addplot+[blue4] table[col sep=comma, x expr=\thisrowno{1}, y expr={\thisrowno{5}*100}]{\dpqtl/validation_outage.csv};

\end{groupplot}
\end{tikzpicture}
	\caption{Empirical outage probability (validation) and dual variables corresponding to data groups of different \gls{csi} quality $\xi_i=\ppilot$ for the training run of fold 1. Nominal outage $\pout=0.1$.\label{fig:training_noqtl_qtl}}
\end{figure}
We compare our proposed adaptive annealing-based constraint in \eqref{eq:qospenalty_stepapx} to the quantile-based penalty in \cite{adamMachineLearningApproach2019, youDataAugmentationBased2021, psomasDesignAnalysisSWIPT2022, yingDeepLearningBasedJoint2024} with the U-G-HBF-GCN model.
Since the dual variables and dual gradients have different scales, we adjust the dual step size $\eta_\upm{d}$, see Tab.~\ref{tab:simdata_param}. %
\par
We observe in Fig.~\ref{fig:4x4_comparison} and Tab.~\ref{tab:4x4_comparison} that the proposed annealing-based constraint yields a slightly better power-outage trade-off, whereas the quantile-based method overfulfills the outage constraint for $\ppilot=\qty{24}{\dB}$.
Fig.~\ref{fig:training_noqtl_qtl} reveals the reason.
Although the relative per-step changes in the dual variables are similar for both methods, the empirical outage probability on the validation data exhibits substantially larger variability for the quantile-based approach, particularly at $\ppilot = 24\unit{\dB}$, which corresponds to the ``easiest'' constraint to satisfy.
Note that the different constraint variants converge to distinct values of $\lambda_d$ due to differences in the scaling of the associated constraint gradient.
We hypothesize that the smoother loss landscape induced by instance averaging in the annealing-based constraint reduces the gradient variance across different data subsets.
This also leads to faster training with $(26\pm 9)\cdot 10^{3}$ steps required for convergence compared to $(34\pm 8)\cdot 10^{3}$, averaged over $15$ runs.
\par
Note that the sensitivity of the proposed adaptive constraint \wrt the value of its hyperparameters is empirically low.
We do not observe a significant difference in model performance when changing the smoothing factor $\eta_\upm{c}$ to $0.002$ and $0.05$.
Thereby, $(29\pm 7)\cdot 10^3$ and $(26\pm 6)\cdot 10^3$ training steps are taken until convergence.
Similarly, increasing the annealing parameter bound $\overline{\beta}_{\upm{c}}$ from $50$ to $250$ does not significantly impact performance.
However, reducing $\overline{\beta}_{\upm{c}}$ to $10$ leads to overfulfillment of the outage constraints by $1.7\%$ on average, indicating that the approximation of the outage in \eqref{eq:qospenalty_stepapx} is not sufficiently close at this value.

\subsection{Generalization w.r.t. Number of Users}\label{sec:4x4_generalization}
\begin{table*}
	\setlength{\tabcolsep}{3.8pt}
	\caption{Results for $M=4\times 4$ antennas with 1 to 4 users for nominal outage $\pout = 0.10$ for \gls{csi} degradation based on different $\ppilot$. U-G-HBF-GCN* is trained with 2-3 users. \label{tab:4x4_generalization}}
	\begin{tabularx}{\linewidth}{cXrrrrrrrr}
		\toprule
		& & \multicolumn{2}{c}{1 user} & \multicolumn{2}{c}{2 users} & \multicolumn{2}{c}{3 users} & \multicolumn{2}{c}{4 users}\\
		\cmidrule(lr){3-4} \cmidrule(lr){5-6} \cmidrule(lr){7-8} \cmidrule(lr){9-10}
		$\ppilot$ & \textbf{Methods} & $\norm{\vp}_1$ & $\estpout$ $[\%]$ & $\norm{\vp}_1$ & $\estpout$ $[\%]$ & $\norm{\vp}_1$ & $\estpout$ $[\%]$ & $\norm{\vp}_1$ & $\estpout$ $[\%]$\\
		\midrule
		\multirow{5}{*}{$10\unit{\dB}$} & G-HBF-perf & $1.30 \pm 0.05$ & $0.00 \pm 0.00$ & $3.02 \pm 0.08$ & $0.00 \pm 0.00$ & $5.18 \pm 0.11$ & $0.00 \pm 0.00$ & $8.23 \pm 0.21$ & $0.00 \pm 0.00$ \\
		& G-HBF-marg & $1.67 \pm 0.06$ & $10.05 \pm 0.04$ & $4.56 \pm 0.13$ & $10.02 \pm 0.05$ & $10.23 \pm 0.26$ & $10.04 \pm 0.03$ & $29.81 \pm 1.78$ & $10.01 \pm 0.03$\\
		& U-G-HBF-GCN  & $1.34 \pm 0.05$ & $9.62 \pm 1.09$ & $3.70 \pm 0.22$ & $9.07 \pm 1.15$ & $7.66 \pm 0.51$ & $9.82 \pm 0.67$ & $16.23 \pm 1.24$ & $11.97 \pm 0.90$\\
		& U-G-HBF-GCN* & $1.31 \pm 0.06$ & $9.65 \pm 1.44$ & $3.53 \pm 0.10$ & $9.49 \pm 0.63$ & $7.72 \pm 0.50$ & $10.79 \pm 0.45$ & $17.35 \pm 1.57$ & $15.94 \pm 1.23$\\
		& U-FDBF-GCN & $1.26 \pm 0.04$ & $8.90 \pm 0.82$ & $3.32 \pm 0.09$ & $8.85 \pm 0.62$ & $6.38 \pm 0.27$ & $9.83 \pm 0.60$ & $11.73 \pm 0.93$ & $11.28 \pm 1.29$\\
		\midrule
		\multirow{5}{*}{$24\unit{\dB}$} & G-HBF-perf & $1.26 \pm 0.06$ & $0.00 \pm 0.00$ & $3.04 \pm 0.10$ & $0.00 \pm 0.00$ & $5.24 \pm 0.03$ & $0.00 \pm 0.00$ & $8.26 \pm 0.19$ & $0.00 \pm 0.00$\\
		& G-HBF-marg & $1.32 \pm 0.06$ & $10.05 \pm 0.04$ & $3.21 \pm 0.10$ & $10.02 \pm 0.05$ & $5.63 \pm 0.03$ & $10.01 \pm 0.03$ & $9.08 \pm 0.21$ & $9.98 \pm 0.05$\\
		& U-G-HBF-GCN  &  $1.04 \pm 0.03$ & $6.74 \pm 0.91$ & $2.43 \pm 0.05$ & $7.90 \pm 0.74$ & $4.19 \pm 0.11$ & $8.73 \pm 0.36$ & $6.82 \pm 0.15$ & $8.99 \pm 0.79$\\
		& U-G-HBF-GCN* &  $1.01 \pm 0.03$ & $7.28 \pm 0.65$ & $2.39 \pm 0.03$ & $8.33 \pm 0.88$ & $4.15 \pm 0.07$ & $9.28 \pm 0.36$ & $6.82 \pm 0.15$ & $9.93 \pm 0.75$
		\\
		& U-FDBF-GCN & $0.98 \pm 0.02$ & $7.14 \pm 0.81$ & $2.20 \pm 0.04$ & $7.89 \pm 0.91$ & $3.61 \pm 0.07$ & $8.69 \pm 0.30$ & $5.41 \pm 0.07$ & $8.91 \pm 0.43$\\
		\bottomrule
	\end{tabularx}
\end{table*}
In this experiment, we investigate the domain adaptation capability of the integrated \glspl{gnn}.
For this, the number of users is randomly sampled from one to four in each training step $\tau$. 
Simultaneously, the number of constraints is quadrupled to $16$ compared to Sec.~\ref{sec:results}-\ref{sec:4x4_comparison}.
\par
We present the results in Tab.~\ref{tab:4x4_generalization} for $\numrf=8$.
Clearly, a single U-G-HBF-GCN model reliably achieves the nominal outage probability within a small error for systems with one to four users and both low and high \gls{csi} quality, thus demonstrating an excellent domain adaptation capability.
An exception are single-user systems with $\ppilot=\qty{24}{\dB}$, where the outage constraint is overfulfilled.
The benchmark G-HBF-marg is again substantially outperformed, and the allocated power is close to G-HBF-perf if $\ppilot$ is high.
We are further interested in extrapolation performance if the tested user count is not contained in the training data set.
U-G-HBF-GCN* in Tab.~\ref{tab:4x4_generalization} denotes models that are trained with only systems of two or three users.
While the achieved power-outage trade-off for two or three users is similar to that of U-G-HBF-GCN, we observe comparable performance for one or four users at $\ppilot=\qty{24}{\dB}$.
When extrapolating to the most challenging case among the tested, four users and $\ppilot=\qty{10}{\dB}$, the outage of U-G-HBF-GCN* only degrades to to $15.94\% \pm 1.23\%$.
This generalization performance can be attributed to our use of a \gls{gcnn}, which is permutation-equivariant \wrt the users and, due to its parametrization, allows seamless application to arbitrary user counts and shift operator dimensions.
Compared to Sec.~\ref{sec:results}-\ref{sec:4x4_comparison}, hybrid \gls{bf} now achieves performance close to fully digital \gls{bf}, except for the case of 4 users.
This occurs at the established threshold $\numrf = 2I$, which likewise applies to hybrid \gls{bf} architectures with more general phase shifting networks \cite[Prop.~2]{sohrabiHybridDigitalAnalog2016}.

\subsection{Combined Statistical \& Instantaneous CSI} \label{sec:exp_statinst}
\begin{table*}
	\setlength{\tabcolsep}{3.8pt}
	\caption{Results for $M=8\times 8$ antennas for combined statistical and instantaneous CSI with DFT feedback quantization. \label{tab:8x8_statinst}}
	\begin{tabularx}{\linewidth}{cXrrrrrrrr}
		\toprule
		& & \multicolumn{2}{c}{$\numfbv \in\uset{4, \dots, 8}$} & \multicolumn{2}{c}{$\numfbv=4$} & \multicolumn{2}{c}{$\numfbv=6$} & \multicolumn{2}{c}{$\numfbv=8$}\\
		\cmidrule(lr){3-4} \cmidrule(lr){5-6} \cmidrule(lr){7-8} \cmidrule(lr){9-10}
		{\#}users $I$ & \textbf{Methods} & $\norm{\vp}_1$ & $\estpout$ $[\%]$ & $\norm{\vp}_1$ & $\estpout$ $[\%]$ & $\norm{\vp}_1$ & $\estpout$ $[\%]$ & $\norm{\vp}_1$ & $\estpout$ $[\%]$\\
		\midrule
		\multirow{4}{*}{$1$} & G-HBF-perf & $1.07 \pm 0.04$ & $0.00 \pm 0.00$ & $1.09 \pm 0.03$ & $0.00 \pm 0.00$ & $1.07 \pm 0.03$ & $0.00 \pm 0.00$ & $1.05 \pm 0.02$ & $0.00 \pm 0.00$\\
		& G-HBF-marg & $2.03 \pm 0.05$ & $10.01 \pm 0.04$ & $3.82 \pm 0.27$ & $10.00 \pm 0.03$ & $1.99 \pm 0.10$ & $10.00 \pm 0.03$ & $1.49 \pm 0.06$ & $9.97 \pm 0.05$\\
		& U-G-HBF-GCN & $2.30 \pm 0.11$ & $1.32 \pm 0.89$ & $4.33 \pm 0.30$ & $2.22 \pm 0.37$ & $2.22 \pm 0.12$ & $1.38 \pm 1.22$ & $1.61 \pm 0.09$ & $1.07 \pm 0.72$\\
		& U-G-HBF-2GCN & $1.05 \pm 0.05$ & $5.33 \pm 0.29$ & $1.22 \pm 0.06$ & $4.26 \pm 0.58$ & $1.05 \pm 0.03$ & $5.67 \pm 0.45$ & $0.95 \pm 0.02$ & $6.53 \pm 0.91$\\
		\midrule
		\multirow{4}{*}{$2$} & G-HBF-perf & $2.60 \pm 0.09$ & $0.00 \pm 0.00$ & $2.61 \pm 0.11$ & $0.04 \pm 0.08$ & $2.61 \pm 0.08$ & $0.00 \pm 0.00$ & $2.59 \pm 0.08$ & $0.01 \pm 0.01$\\
		& G-HBF-marg & $4.41 \pm 0.35$ & $10.02 \pm 0.03$ & $6.55 \pm 0.36$ & $10.00 \pm 0.04$ & $4.38 \pm 0.19$ & $10.01 \pm 0.05$ & $3.67 \pm 0.16$ & $10.01 \pm 0.05$\\
		& U-G-HBF-GCN & $4.35 \pm 0.13$ & $8.78 \pm 0.79$ & $6.30 \pm 0.63$ & $9.20 \pm 1.73$ & $4.22 \pm 0.25$ & $8.66 \pm 0.68$ & $3.43 \pm 0.15$ & $9.63 \pm 1.30$\\
		& U-G-HBF-2GCN  & $3.15 \pm 0.08$ & $9.64 \pm 0.68$ & $3.95 \pm 0.15$ & $9.91 \pm 1.00$ & $3.09 \pm 0.11$ & $9.85 \pm 1.09$ & $2.66 \pm 0.06$ & $9.68 \pm 0.81$\\
		\bottomrule
	\end{tabularx}
\end{table*}
Finally, we increase $M$ to $8\times8$, $\numrf$ to $16$ and sample $\gamma_i$ in the interval $[10\unit{\decibel},20\unit{\decibel}]$.
In addition, we consider \gls{dft} quantization with $D_\upm{c}=4$ groups, where $\numfbv=4$, $\numfbv=6$, $\numfbv=8$ and $\numfbv \in \uset{4,\dots, 8}$ uniformly, respectively.
\par
Furthermore, for each system instance $\ssample$, we sample $32$ instantaneous channel coefficients $(\vh_{i,s})_{s=1}^{32}$ and compute $\hbR_i$ as the covariance estimate of $(\hvh_{i,s})_{s=1}^{32}$.
The estimated statistical \gls{csi} $(\hbR_i)_i$ is used to choose the analog beamforming matrix $\bA$, while an additional uplink problem $\digbfmapul$ is solved to obtain digital precoders $\bB$ for each corresponding sampled coefficient $(\hvh_{i, s})_i$ individually.
In addition, denoted as U-G-HBF-2GCN, an unrolled \gls{dn} employing separate \glspl{gcnn}, and thus separate virtual channel covariances $\hbPsi_{i,j}$, for the analog beams and digital precoders is evaluated.
This mimics the slower reconfiguration capability of analog components, which therefore rely on the channel statistics, while the quickly configurable digital precoders leverage instantaneous information.
The outage probability is evaluated \wrt instantaneous \gls{csi}.%
\par
Tab.~\ref{tab:8x8_statinst} shows that the proposed architecture U-G-HBF-2GCN significantly outperforms the baseline G-HBF-marg for $2$ users and all investigated values of $\numfbv$, thereby clearly improving the allocated power if $\numfbv$ increases.
The outage constraint for a single user is overfulfilled while the corresponding dual variables remain zero during training.
This behavior may be caused by insufficient degrees of freedom in the \gls{gcnn} architecture or the channel mapping in \eqref{eq:est_ch_model}.
In comparison, U-G-HBF-GCN does not significantly outperform G-HBF-marg, demonstrating the benefit of separate \glspl{gcnn} for statistical and instantaneous \gls{csi}.

\section{Conclusion and Outlook}\label{sec:conclusion}
We propose a novel deep-learning-based model architecture for hybrid downlink beamforming under probabilistic \gls{qos} constraints.
To compute the model gradient, we derive a computationally efficient gradient for the widely applicable constrained uplink/downlink beamforming problem and establish sufficient conditions for its existence.
Experiments show that the proposed method can satisfy multiple outage constraints while outperforming the benchmarks in terms of total allocated power.
In addition, the proposed methods exhibit an excellent capability of generalizing to different levels of imperfect channel state information and numbers of users, even when extrapolating to an unseen number of users.
With a sufficient number of \gls{rf} chains, hybrid beamforming attains performance comparable to a reduced deep network architecture designed for fully digital beamforming.
We further demonstrate that an alternative annealing-based approach for enforcing probabilistic constraints reduces training time by smoothing the loss landscape, while also yielding slightly improved deep network performance. %
\par
Future work could explore virtual channel mappings with higher degrees of freedom, pruning techniques for the analog codebook, alternative hybrid architectures as well as multi-carrier and multi-cell systems.
Furthermore, reinforcement learning could be investigated to eliminate the requirement for exact channel state information during training.

\appendices

\section{Jacobi Matrices of the Baseband BF Problem}\label{sec:apx_jacobi}
In the following, the Jacobi matrices corresponding to the \gls{snle} \eqref{eq:realval_kkt} for the digital \gls{bf} problem are detailed.
We adopt the notation of \cite{hjorungnesComplexValuedMatrixDerivatives2011} that $\vect{\bPsi} = \bGamd \uhvpsid + \bGamr\uhvpsir + \ju \bGami\uhvpsii$, where $\uhvpsid\in\Rset^{\numrf}$, $\uhvpsir\in\Rset^{\numrf(\numrf-1)/2}$ and $\uhvpsii\in\Rset^{\numrf(\numrf-1)/2}$ are the diagonal, real upper triangular and imaginary upper triangular of $\hbPsi\in\Rset^{\numrf\times \numrf}$, respectively, and $\bGamd$, $\bGamr$ and $\bGami$ are the associated mapping matrices.
Further, note that we can rewrite $\uvg_{\upm{b}}$ in \eqref{eq:realval_kkt} as
\begin{align}
	\uvg_{\upm{b}}(\ubvzeta, \uhvpsi) 
		&= \bmat*{ \pdel*{\ubLambda_1(\vq, \uhvpsi)\ubvb_1}^{\!\T}~ \cdots~ \pdel*{\ubLambda_I(\vq, \uhvpsi)\ubvb_I}^{\!\T}}^\T
\end{align}
where we define
\begin{align}
	\ubLambda_i(\vq, \hvpsi) &= \uhbPsi_{0, i} - \gamma_i^{-1} q_i\uhbPsi_{i,i} + \sum_{j\neq i} q_j \uhbPsi_{j,i} \allowdisplaybreaks\\
	\text{and}\quad \uhbPsi_{i,j} &= \begin{bmatrix}
			\Re(\hbPsi_{i,j}) & \!\!-\matel{\Im(\hbPsi_{i,j})}{:, \neg \numrf}\\
			\matel{\Im(\hbPsi_{i,j})}{ \neg \numrf, :} & \!\! \matel{\Re(\hbPsi_{i,j})}{ \neg \numrf, \neg \numrf}
		\end{bmatrix}\!\!.
\end{align}

In the following, the arguments $(\ubvzeta, \uhvpsi)$ are omitted for readability.
\par
\subsubsection{Left-hand side} 
The left-hand side is decomposed as
\begin{align}
	\jcby[\ubvzeta]\uvg = 
		\begin{bmatrix}
			\jcby[\ubvb] \uvg_{\upm{b}} & \jcby[\vq] \uvg_{\upm{b}}\\
			\jcby[\ubvb] \uvg_{\upm{q}} & \jcby[\vq] \uvg_{\upm{q}}
		\end{bmatrix}.
\end{align}
For the upper left block, we have
\begin{align}
	\jcby[\ubvb]\uvg_{\upm{b}} &= \Blkdiag*{\jcby[\ubvb_1] \uvg_{\upm{b},1}, \dots, \jcby[\ubvb_I] \uvg_{\upm{b},I}}, \\
	\jcby[\ubvb_i] \uvg_{\upm{b},i} &=  \ubLambda_i(\vq, \uhvpsi).
\end{align}
The upper left block is itself composed of sub-blocks
\begin{align}
	\jcby[\vq] \uvg_{\upm{b}} &= \begin{bmatrix}
		\jcby[q_j] \uvg_{\upm{b}_i} \stackrel{j}{\longrightarrow}\\
		{\scriptstyle i}\!\big\downarrow~~~~~\ddots
	\end{bmatrix},  \\[1mm]
	\text{where } \jcby[q_j] \uvg_{\upm{b},i} &= \begin{cases}
		-\gamma_i^{-1}\uhbPsi_{i, i}\ubvb_i & \text{for $j=i$},\\
		\hbPsi_{j, i}\ubvb_i & \text{for $j\neq i$.}
	\end{cases}
\end{align}
Leveraging symmetry, the lower left block is constructed as
\begin{align}
	\jcby[\ubvb] \uvg_{\upm{q}} &= 2\Diag*{\vq}\jcby[\vq] \uvg_{\upm{b}}^\T, \label{eq:derivative_symm}
\end{align}
while for the lower right block, we have
\begin{align}
	\jcby[\vq] \uvg_{\upm{q}} \! &= \Diag{\jcby[q_1] \uvg_{\upm{q},1}, \cdots, \jcby[q_I] \uvg_{\upm{q},I}},\\
	\text{with } \jcby[q_i] \uvg_{\upm{q},i} \! &= 1 - \gamma_i^{-1} \ubvb_i^\T \uhbPsi_{i,i}\ubvb_i  + \sum_{j\neq i} \ubvb_j^\T \uhbPsi_{i,j} \ubvb_j.
\end{align}
Note that if the problem is feasible, the \gls{sinr} is exactly achieved \cite{bocheGeneralDualityTheory2002}, thus, $\jcby[\vq] \uvg_{\upm{q}}|_{\vzeta=\vzeta^\opt} = \nullmat$.
\par
\subsubsection{Right-hand side}
Decompose the right-hand side as 
\begin{align}
	\jcby[\uhvpsi_{i,j}] \uvg = \begin{bmatrix}
		\jcby[\uhvpsi_{i,j}]\uvg_{\upm{b}}\\
		\jcby[\uhvpsi_{i,j}]\uvg_{\upm{q}}
	\end{bmatrix}.
\end{align}
The upper block is constructed as
\begin{align}
	\jcby[\uhvpsi_{i,j}]\! \uvg_{\upm{b}} = \bmat*{{\scriptstyle k}\!\big\downarrow ~ \jcby[\uhvpsi_{\upm{d}, (i,j)}]\! \uvg_{\upm{b},k} ~ \jcby[\uhvpsi_{\upm{re}, (i,j)}]\! \uvg_{\upm{b},k} ~ \jcby[\uhvpsi_{\upm{im}, (i,j)}]\! \uvg_{\upm{b},k}},
\end{align}
where, for $\alpha \in \uset{\upm{d}, \upm{re}}$, we have ($\idmat$ has size $\numrf \times \numrf$)
\begin{align}
	\jcby[\uhvpsi_{\alpha,(i,j)}] \! \uvg_{\upm{b},k} \!&= \!\begin{cases}
		-\dfrac{q_i}{\gamma_i} \! \begin{bmatrix}
			(\Re(\bvb_i)^\T \! \kron\idmat)\bGamein \\  \matel{\Im(\bvb_i)^\T \! \kron\idmat}{:, \neg \numrf}\bGamein
		\end{bmatrix} & \!\!\!\! \!\!\!\begin{array}{l}
		\text{for~} \\ ~i = j = k,
		\end{array}\\
		q_i \! \begin{bmatrix}
			(\Re(\bvb_i)^\T \! \kron\idmat)\bGamein \\  \matel{\Im(\bvb_i)^\T \! \kron \idmat}{:, \neg \numrf}\bGamein
		\end{bmatrix} & \!\!\!\! \!\!\!\begin{array}{l}
			\text{for~}i\neq 0,\\~ i\neq j, j=k,
		\end{array}\\
		\begin{bmatrix}
			(\Re(\bvb_i)^\T \! \kron\idmat)\bGamein \\  \matel{\Im(\bvb_i)^\T \! \kron \idmat}{:, \neg \numrf}\bGamein
		\end{bmatrix} & \!\!\!\! \!\!\!\begin{array}{l}
			\text{for~} i= 0,\\~ i\neq j, j=k,
		\end{array}\\
		\nullmat & \!\!\!\! \text{otherwise.}
	\end{cases}
\end{align}
For $\alpha = \upm{im}$, we have
\begin{align}
	\jcby[\uhvpsi_{\upm{im},(i,j)}] \! \uvg_{\upm{b},k} \! &= \! \begin{cases}
		-\dfrac{q_i}{\gamma_i} \! \begin{bmatrix}
			(-\Im(\bvb_i)^\T \! \kron\idmat)\bGamein \\  \matel{\Re(\bvb_i)^\T \! \kron\idmat}{:, \neg \numrf}\bGamein
		\end{bmatrix} & \!\!\!\! \!\!\!\!\begin{array}{l}
		\text{for~} \\ ~i = j = k,
		\end{array}\\
		q_i \! \begin{bmatrix}
			(-\Im(\bvb_i)^\T \! \kron\idmat)\bGamein \\  \matel{\Re(\bvb_i)^\T \! \kron\idmat}{:, \neg \numrf}\bGamein
		\end{bmatrix} & \!\!\!\! \!\!\!\!\begin{array}{l}
			\text{for~} i\neq 0,\\~ i\neq j, j=k,
		\end{array}\\
		\begin{bmatrix}
			(-\Im(\bvb_i)^\T \! \kron\idmat)\bGamein \\  \matel{\Re(\bvb_i)^\T \! \kron\idmat}{:, \neg \numrf}\bGamein
		\end{bmatrix} & \!\!\!\! \!\!\!\!\begin{array}{l}
		\text{for~} i=0,\\ ~j=k,
		\end{array}\\
		\nullmat & \!\!\!\! \text{otherwise.}
	\end{cases} 
\end{align}
For the lower block, we similarly have
\begin{align}
	\jcby[\uhvpsi_{i,j}]\! \uvg_{\upm{q}} = \bmat*{{\scriptstyle k}\!\big\downarrow ~\jcby[\uhvpsi_{\upm{d}, (i,j)}]\! \uvg_{\upm{q},k} ~ \jcby[\uhvpsi_{\upm{re}, (i,j)}]\! \uvg_{\upm{q},k} ~ \jcby[\uhvpsi_{\upm{im}, (i,j)}]\! \uvg_{\upm{q},k}}, 
\end{align}
where for $\alpha \in \uset{\upm{d}, \upm{re}}$
\begin{align}
	\jcby[\uhvpsi_{\alpha,(i,j)}] \uvg_{\upm{q},k} &= \begin{cases}
		-\dfrac{q_i}{\gamma_i}\Re ( \vect{\bvb_i \bvb_i^\He}) \bGamein & \text{for $i=j=k$}\\
		q_i\Re(\vect{\bvb_j \bvb_j^\He})\bGamein & \text{for $i=k,i\neq j$}\\
		0 & \text{otherwise.}
	\end{cases}
\end{align}
Lastly, for $\alpha = \upm{im}$, we have
\begin{align}
	\jcby[\uhvpsi_{\upm{im},(i,j)}] \uvg_{\upm{q},k} &= \begin{cases}
		-\dfrac{q_i}{\gamma_i}\Im ( \vect{\bvb_i \bvb_i^\He}) \bGamein & \text{for $i=j=k$}\\
		q_i\Im(\vect{\bvb_j \bvb_j^\He})\bGamein & \text{for $i=k,i\neq j$}\\
		0 & \text{otherwise.}
	\end{cases}
\end{align}
\subsubsection{Beamformer normalization}
The associated Jacobi matrix with the beamformer normalization can be found as
\begin{align}
	&\jcby[\ubvzeta] \vf_\upm{BN}(\ubvzeta) \nonumber \\[-1mm]
	&~ = \Blkdiag*{\!\jcby[\ubvb_1] \breve{\vf}_\upm{BN}(\ubvb_1), \dots, \jcby[\ubvb_I] \breve{\vf}_\upm{BN}(\ubvb_I), \idmat_{I\times I}\!},\\
	&\jcby[\ubvb_i] \breve{\vf}_\upm{BN}(\ubvb_i) \nonumber\\[-1mm]
	&~ = \frac{1}{\norm{\bvb_i}^{\frac{3}{2}}} \!\left( \!\!\norm{\bvb_i}_2^2\idmat \!-\! \begin{bmatrix}
		\Re(\bvb_i)\Re(\bvb_i)^\T & \!\Re(\bvb_i)\matel{\Im(\bvb_i)}{\neg \numrf}^\T\\
		\Im(\bvb_i)\Re(\bvb_i)^\T & \!\Im(\bvb_i)\matel{\Im(\bvb_i)}{\neg \numrf}^\T
	\end{bmatrix}\!\right)\!.
\end{align}
\section{Identification of the Derivative}\label{sec:apx_derivative_proof}

\begin{proof}[Proof of \textbf{Theorem}~\ref{th:dlbf_derivative}] 
	\par
	We first argue that a solution of Problem~\eqref{eq:problem_ul_inexactcsi_digital} with $\matel{\Im(\bvb_i^\opt)}{\numrf} = 0$ for all $i$ corresponds to a critical point of $\uvg(\ubvzeta, \hvpsi)$.
	The semidefinite relaxation of the dual downlink problem corresponding to Problem~\eqref{eq:problem_ul_inexactcsi_digital} to the convex formulation in \eqref{eq:problem_dl_convex} is exact, \cf \cite{gershmanConvexOptimizationBasedBeamforming2010}, \cite[Th.~3.2]{huangRankConstrainedSeparableSemidefinite2010}.
	Since \gls{kkt} conditions are necessary and sufficient in this case \cite[Ch.~5.5.3]{boydConvexOptimization2004}, it follows from Lemma~\ref{lem:kkt_equivalence} that any solution $\oset{\vb_i^\opt, q_i^\opt}_i$ of Problem~\eqref{eq:problem_ul_inexactcsi_digital} has a corresponding dual feasible solution $\oset{\bvb_i^\opt, q_i^\opt}_i = \oset{\sqrt{p_i^\opt}\vb_i^\opt, q_i^\opt}_i$ to the \gls{snle} \eqref{eq:kkt_nonconvex_psd}-\eqref{eq:kkt_nonconvex_sinr}, where the downlink power $(p_i^\opt)_i$ is given by \eqref{eq:dlpow_calc_apx}.
	Due to arbitrariness of the baseband beamformers' phase rotations, we can assume without loss of generality that $\matel{\Im(\bvb_i^\opt)}{\numrf} = 0$ for all $i$.
	We thus conclude from the definition of $\uvg(\ubvzeta, \hvpsi)$ in \eqref{eq:realval_kkt} that $\uvg(\ubvzeta^\opt, \hvpsi) = \nullvec$, where $\ubvzeta^\opt$ is defined analogously to \eqref{eq:def_realval_primdual_var}.
	
	\par
	We now show by contradiction that $\jcby[\ubvzeta] \uvg(\ubvzeta^\opt, \uhvpsi)$ is not rank-deficient.
	Suppose $\jcby[\ubvzeta] \uvg(\ubvzeta^\opt, \uhvpsi)$ is rank-deficient, then $\jcby[\ubvzeta]{\uvg(\ubvzeta^\opt, \uvpsi)} \ubvzeta' = \nullvec$ for some $\ubvzeta' = \bmat{(\ubvb')^\T~~(\vq')^\T}^\T \neq \nullvec$.
	Leveraging the symmetry in \eqref{eq:derivative_symm}, this can be written as (the argument $\uhvpsi$ is omitted for readability)
	\begin{align}
		\jcby[\ubvb] \uvg_{\upm{b}} (\ubvzeta^\opt) \ubvb' + \frac{1}{2} \jcby[\vb] \uvg_{\upm{q}}^\T (\ubvzeta^\opt) \Diag{\vq^\opt}^{-1} \vq'&= \nullvec, \label{eq:rankdef_eq1}\\
		\jcby[\ubvb] \uvg_{\upm{q}} (\ubvzeta^\opt) \ubvb' &= \nullvec. \label{eq:rankdef_eq2}
	\end{align}
	First, consider the case $\ubvb' \neq \nullvec$.
	Multiplying $(\ubvb')^\T$ from the left to \eqref{eq:rankdef_eq1} and $(\vq')^\T \Diag{\vq^\opt}^{-1}$ to \eqref{eq:rankdef_eq2} yields
	\begin{align}
		\!(\ubvb')^\T \!\jcby[\ubvb] \uvg_{\upm{b}} (\ubvzeta^\opt)\ubvb' \!+\! \frac{1}{2}\ubvb^\T \! \jcby[\vq] \uvg_{\upm{q}}^\T (\ubvzeta^\opt) \Diag{\vq^\opt}^{-1} \vq'\!&= 0, \!\! \label{eq:rankdef_eq1step}\\
		\!(\vq')^\T \!\Diag{\vq^\opt}^{-1} \jcby[\ubvb] \uvg_{\upm{q}} (\ubvzeta^\opt) \ubvb' \!&= 0, \!\! \label{eq:rankdef_eq2step}
	\end{align}
	and thus $(\ubvb')^\T \jcby[\ubvb] \uvg_{\upm{b}} (\ubvzeta^\opt)\ubvb' = 0$ by substituting \eqref{eq:rankdef_eq2step} into \eqref{eq:rankdef_eq1step}.
	Since $\bLambda_i(\vq^\opt) \succeq\nullmat$ (dual feasible solution), we know that $\ubLambda_i(\vq^\opt) \succeq\nullmat$ and that $\jcby[\ubvb] \uvg_{\upm{b}}(\ubvzeta^\opt) \succeq \nullmat$ as well.
	Therefore, we deduce that $\ubvb' \in \nullspace{\jcby[\ubvb] \uvg_{\upm{b}}(\ubvzeta^\opt)}$, where $\nullspace{\cdot}$ denotes the matrix kernel.
	\par
	Let us now rewrite $\ubvb_i^\opt$ as $\ubvb_i^\opt = s_i^\opt\uvb_i^\opt$ with normalized real-valued baseband beamformers $\norm{\uvb_i^\opt}_2 = 1$ and scale $s_i^\opt > 0$ for all $i=1,\dots,I$.
	Further define $\hbC \pdel*{\oset{\vx_i}_i, \oset{\vy_i}_i} \in \Rset^{I \times I}$ as
	\begin{equation}
		\matel{\hbC \pdel*{\oset{\vx_i}_i, \oset{\vy_i}_i}}{j,k} = \begin{cases}
			\gamma_j^{-1} \vx_j^\T \uhbPsi_{j,j}\vy_j & \text{for $j=k$,}\\
			-\vx_k^\T \uhbPsi_{j,k}\vy_k & \text{for $j\neq k$,}
		\end{cases}
	\end{equation}
	where $\oset{\vx_i}_i$ and $\oset{\vy_i}_i$ are conformable arguments.
	Suppose that $\ubvb_i' = \pm s_i' \uvb_i^\opt$ for some $s_i' > 0$ for all $i$, which immediately implies $\ubvb' \in \nullspace{\jcby[\ubvb] \uvg_{\upm{b}}(\ubvzeta^\opt)}$.
	Additionally, rewrite the left-hand side of \eqref{eq:rankdef_eq2} as
	\begin{equation}
		\jcby[\ubvb] \uvg_{\upm{q}} (\vzeta^\opt) \ubvb' = - 2\Diag{\vq^\opt} \hbC \big(\oset{\uvb_i^\opt}_i, \oset{\uvb_i^\opt}_i\big) (\bs^\opt \hmul \bs'). \label{eq:proof_derivative_rewrite}
	\end{equation}
	We can identify the coupling matrix $\hbC \big(\oset{\uvb_i^\opt}_i, \oset{\uvb_i^\opt}_i\big) = \hbC_{\upm{C}}(\bB^\opt, \bA)$ (see definition in \eqref{eq:apx_coupling_matrix}), which has full rank since $\bB^\opt$ is a solution to the baseband optimization problem \cite{bocheGeneralDualityTheory2002}.
	It follows that \eqref{eq:proof_derivative_rewrite} can only be $\nullvec$ if $\bs' = \nullvec$.
	\par
	Therefore, suppose instead $\ubvb_i' = s_i' \uvb_i'$, where $\uvb_i' \neq \pm \uvb_i^\opt$ for some $i=k$.
	Since we require that $\ubvb' \in \nullspace{\jcby[\ubvb] \uvg_{\upm{b}}(\ubvzeta^\opt)}$, it is implied that $\dim{\nullspace{\ubLambda_k(\vq^\opt)}} \geq 2$.
	As such, there also exists a $\ubvzeta''$ with $\uvb_k'' \neq \uvb_k^\opt$ and $\vq'' = \vq^\opt$ such that $\ubLambda_k(\vq^\opt)\ubvb_k''$ and thus $\uvg_{\upm{b}} (\ubvzeta'') = \nullvec$.
	The complementarity condition $\uvg_{\upm{q}} (\ubvzeta'')$ can be expressed as
	\begin{align}
		&\uvg_{\upm{q}} (\ubvzeta'') = \Diag{\vq^\opt} \left(\onevec \!-\! \hbC \big(\oset{\uvb_i''}_i, \oset{\uvb_i''}_i\big) (\bs'' \!\hmul \bs'')\right) \!. \label{eq:proof_critpoint}
	\end{align}
	Due to the lower semi-continuity of the rank, if $\ubvzeta''$ is chosen sufficiently close to $\ubvzeta^\opt$, then $\hbC \big(\oset{\uvb_i''}_i, \oset{\uvb_i''}_i\big)$ still has rank $I$.
	Thus, we can find $\bs'' > \nullvec$ such that $\uvg_{\upm{q}} (\ubvzeta'') = \nullvec$ holds and $\ubvzeta''$ is a distinct primal and dual feasible critical point.
	In other words, $\ubvzeta''$ constructs a distinct solution $\oset{\breve{\bB}_i'', q_i''}_i$ of Problem~\eqref{eq:problem_dl_convex} different from $\oset{\breve{\bB}_i^\opt, q_i^\opt}_i = \oset{\bvb_i\bvb_i^\He, q_i^\opt}_i$.
	Since this violates Assumption~\ref{asm:dlbf_unique_sol}, we conclude that $\ubvb' = \nullvec$.
	\par
	Consider now $\vq'$.
	As $\ubvb' = \nullvec$, we have $\jcby[\vq] \uvg_{\upm{b}} (\ubvzeta^\opt) \vq'= \nullvec$.
	After left-multiplying $\Blkdiag{\ubvb_1^\opt, \dots, \ubvb^\opt_I}^\T$ on both sides, the following identity can be verified:
	\begin{align}
		\nullvec &= \Blkdiag{\ubvb_1^\opt, \dots, \ubvb^\opt_I} \jcby[\vq]\uvg_{\upm{b}} (\ubvzeta^\opt) \vq' \nonumber\\
		&= - \Diag{\bs^\opt \hmul \bs^\opt} \hbC^\T \big(\oset{\uvb_i^\opt}_i, \oset{\uvb_i^\opt}_i\big) \vq'. \label{eq:proof_qcase_identity}
	\end{align}
	Since both leading matrices in the second line of \eqref{eq:proof_qcase_identity} have full rank, we conclude that $\vq' = \nullvec$.
	Therefore, $\ubvzeta'= \nullvec$ is the only solution to
	$\jcby[\ubvzeta]{\uvg(\ubvzeta^\opt, \uvpsi)} \ubvzeta' = \nullvec$ and $\jcby[\ubvzeta]{\uvg(\ubvzeta^\opt, \uvpsi)}$ has full rank.
	\par
	As $\ubvzeta^\opt$ constructs a solution of Problem~\eqref{eq:problem_ul_inexactcsi_digital} with $\matel{\Im(\bvb_i^\opt)}{\numrf} = 0$ and is a critical point of $\uvg(\ubvzeta, \uvpsi)$, and additionally $\jcby[\ubvzeta]{\uvg(\ubvzeta^\opt, \uvpsi)}$ is nonsingular, the implicit function theorem \cite[Th.~8.2]{amannAnalysisII2008} is applicable.
	Application of the chain rule due to the composition with the normalization in $\eqref{eq:bfnormfun}$ immediately yields \eqref{eq:implicit_derivative}. 
	\end{proof}
	\begin{remark}
		The nonsingularity of $\jcby[\ubvzeta]{\uvg(\ubvzeta^\opt, \uvpsi)}$ implies that $\ubvzeta^\opt$ is an \underline{isolated} critical point of $\uvg(\ubvzeta^\opt, \uvpsi)$ \cite[Corollary 7.8]{amannAnalysisII2008}.
		The equations that have been removed in \eqref{eq:realval_kkt} are indeed redundant.
	\end{remark}
	
	\section{Pseudocode of the Training Procedure}\label{sec:training_proc}
	\begin{algorithm}[h]
		\small
		\caption{}\label{alg:approx_pi}
		\begin{algorithmic}
			\State \textbf{input} $\dataset_{\upm{train}}$, $\dataset_{\upm{val}}$, $\vtheta\idx{0}$, $\eta_{\upm{p}}$, $\eta_{\upm{d}}$, $\eta_{\upm{a}}$, $\overline{\beta}_\upm{c}$,  $\eta_{\upm{c}}$, $\tau_{\upm{pat}, 1}$, $\tau_{\upm{pat}, 2}$, $\tau_\upm{val}$, $\overline{\lambda}_d$
			\State $\tau, \kappa \gets 0$, $\vlambda\idx{0} \gets \nullvec$
			\State $\kappa_{\upm{pat},1} \gets \tau_{\upm{pat}, 1} // \tau_\upm{val}$, $\kappa_{\upm{pat},2} \gets \tau_{\upm{pat}, 2} // \tau_\upm{val}$
			\While{True}
			\If{$\tau \bmod \tau_\upm{val} = 0$}
			\State $J_\upm{cm}\idx{\kappa} \gets$ \eqref{eq:stopping_metric} with $\overline{\lambda}_d$
			\If{$\min\uset*{J_\upm{cm}\idx{\kappa'}}_{\!\kappa'=0}^{\!\kappa} < \min\uset*{J_\upm{cm}\idx{\kappa'}}_{\!\kappa' = \max\{0, \kappa - \kappa_{\upm{pat},1}\}}^{\!\kappa}$ }
			\State $\eta_{\upm{p}} \gets \eta_{\upm{a}}\eta_{\upm{p}}$, $\eta_{\upm{d}} \gets \eta_{\upm{a}}\eta_{\upm{d}}$ %
			\State $\eta_{\upm{a}} \gets 1$ %
			\EndIf
			\If{$\min\uset*{J_\upm{cm}\idx{\kappa'}}_{\!\kappa'=0}^{\!\kappa} < \min\uset*{J_\upm{cm}\idx{\kappa'}}_{\!\kappa' = \max\{0, \kappa - \kappa_{\upm{pat},2}\}}^{\!\kappa}$ }
			\State \textbf{break}
			\EndIf
			\State $\kappa \gets \kappa + 1$
			\EndIf
			\State Sample minibatch $\breve{\dataset} \sim \dataset_{\upm{train}}$
			 \State $\vtheta\idx{\tau+1} \gets \vtheta\idx{\tau} - \eta_{\upm{p}} \operatorname{AdamW}\pdel*{\nabla_{\vtheta} J(\breve{\dataset},
			 \vtheta\idx{\tau}, \vlambda\idx{\tau}), \vtheta\idx{\tau}}$ \cite{loshchilovDecoupledWeightDecay2019}
			 \Statex ~~~~~~~~~~~~~~~~~ with adaptive gradient clipping \cite{seetharamanAutoclipAdaptiveGradient2020},
			 
			\State $\vlambda\idx{\tau+1} \gets \bdel*{\vlambda\idx{\tau} + \eta_{\upm{d}}\nabla_{\vlambda} J(\breve{\dataset}, \vtheta\idx{\tau}, \vlambda\idx{\tau})}_0^\infty$
			\State $\beta_{\upm{c}}\idx{\tau+1} \gets$ \eqref{eq:qospenalty_stepapx_step} using $\breve{\dataset}$ with $\overline{\beta}_\upm{c}$, $\eta_{\upm{c}}$
			\State $\tau \gets \tau + 1$
			\EndWhile
			\State \textbf{return} $\vtheta\idx{\tau}$
		\end{algorithmic}
	\end{algorithm}

\bibliographystyle{IEEEtran}
\bibliography{journal_refs}

\begin{IEEEbiography}[{\includegraphics[width=1in,height=1.25in,clip,keepaspectratio]{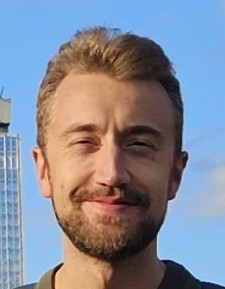}}]{Lukas Schynol}~(Graduate Student Member, IEEE) received the M.Sc. degree in 2021 in electrical engineering from Technical University of Darmstadt, Darmstadt, Germany, where he is currently working towards the Ph.D. degree. 
His research interests include resource allocation and anomaly detection in wireless networks using model-aided deep learning. 
He received the best student paper award (2nd place) at the IEEE International Workshop on Computational Advances in Multi-Sensor Adaptive Processing (CAMSAP) in 2023.
\end{IEEEbiography}

\begin{IEEEbiography}[{\includegraphics[width=1in,height=1.25in,clip,keepaspectratio]{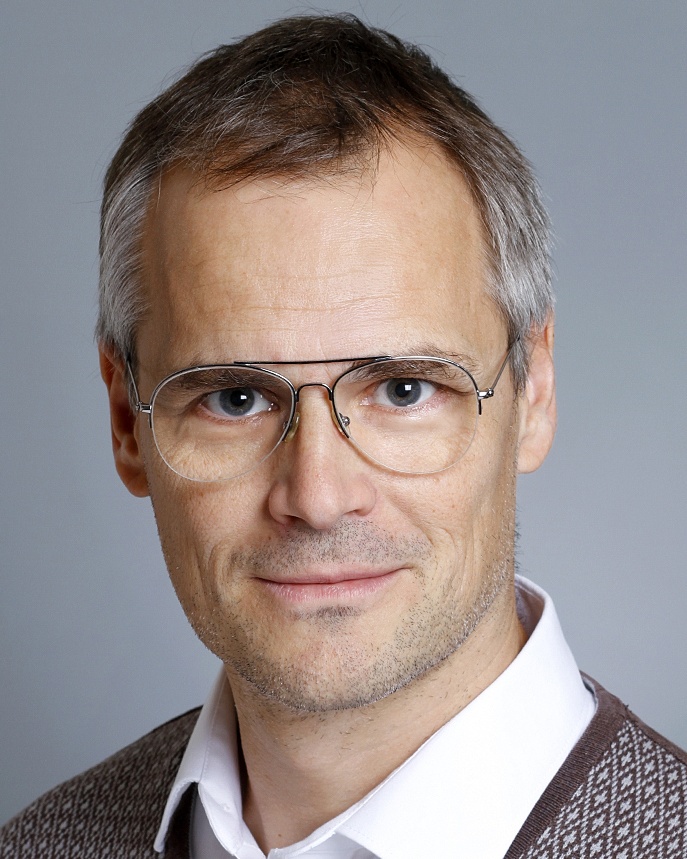}}]{Marius Pesavento}~(Senior Member, IEEE) received the Dipl.-Ing.~and Ph.D.~degrees from Ruhr-Uni\-ver\-si\-tät Bochum, Germany, in 1999 and 2005, respectively, and the M.Eng. degree from McMaster University, Hamilton, ON, Canada, in 2000. 
From 2005 to 2008, he was a Research Engineer with two start-up companies. 
He became an Assistant Professor in 2010 and a Full Professor (W3) of Communication Systems in 2013 in the Department of Electrical Engineering and Information Technology, Technische Universität Darmstadt, Germany. 
His research interests include statistical and sensor array signal processing, MIMO communications, optimization, and model-assisted deep learning. 
	
He is the Regional Director-at-Large for Region 8 (2025–2027) and a member of the IEEE Signal Processing Society Board of Governors. 
He is the Editor-in-Chief (2026–2028) of the IEEE Open Journal of Signal Processing (Deputy Editor-in-Chief, 2025; Senior Area Editor, 2019–2024). 
He was an Associate Editor of the IEEE Transactions on Signal Processing (2012–2016) and a Subject Editor of the EURASIP Journal on Signal Processing (2024–2025; Handling Editor, 2011–2023).
He is a member of the ``Signal Processing Theory and Methods'' Technical Committee (since 2021), served as Vice-Chair of the IEEE SPS “Sensor Array and Multichannel Signal Processing” Technical Committee in 2025 (member, 2012–2017), and is Past Chair of the EURASIP ``Signal Processing for Multisensor Systems'' Technical Area Committee.
\end{IEEEbiography}

\end{document}